\documentclass[12pt]{iopart}

\usepackage[table]{xcolor}

\usepackage{graphicx} 
\expandafter\let\csname equation*\endcsname\relax
\expandafter\let\csname endequation*\endcsname\relax
\usepackage{graphicx,amsmath,amssymb,tikz,psfrag}
\usepackage{bm}
\usetikzlibrary{math}
\usepackage{amsthm}
\usepackage{url}
\DeclareMathAlphabet{\mathmybb}{U}{bbold}{m}{n}
\newcommand{\1}{\mathmybb{1}}

\DeclareMathOperator*{\argmin}{argmin}

\newcommand{\bsJ}{\mathbf{J}}
\newcommand{\bsV}{\mathbf{V}}
\newcommand{\bsm}{\mathbf{m}}
\newcommand{\bsg}{\boldsymbol{g}}
\newcommand{\bsd}{\mathbf{d}}
\newcommand{\bsx}{\mathbf{x}}
\newcommand{\bss}{\mathbf{s}}
\newcommand{\bsz}{\mathbf{z}}
\newcommand{\bsu}{\mathbf{u}}
\newcommand{\bsc}{\mathbf{c}}

\newcommand{\bsr}{\mathbf{r}}
\newcommand{\bsC}{\mathbf{C}}
\newcommand{\bsy}{\mathbf{y}}

\definecolor{my_gray}{HTML}{DBDBDB}
\usepackage{booktabs, makecell, multirow}

\newcommand{\bseta}{\boldsymbol{\eta}}

\newtheorem{theorem}{Theorem}
\newtheorem{assumption}{Assumption}

\begin{document}

\title[Optimizing Quantitative Photoacoustic Imaging Systems]{Optimizing Quantitative Photoacoustic Imaging Systems: The Bayesian Cram{\'e}r-Rao Bound Approach}
\author{Evan Scope Crafts,$^{1}$ Mark A. Anastasio,$^{2}$ and Umberto Villa$^{1}$}
\address{$^{1}$Oden Institute for Computational Engineering and Sciences, The University of Texas, Austin, TX 78712, USA \\
$^{2}$Department of Bioengineering, University of Illinois Urbana-Champaign, Urbana, IL 61801, USA}
\vspace{10pt}
\begin{indented}
\item[] Submitted \today
\end{indented}

\begin{abstract}

Quantitative photoacoustic computed tomography (qPACT) is an emerging medical imaging modality that carries the promise of high-contrast, fine-resolution imaging of clinically relevant quantities like hemoglobin concentration and blood-oxygen saturation. However, qPACT image reconstruction is governed by a multiphysics, partial differential equation (PDE) based inverse problem that is highly non-linear and severely ill-posed. Compounding the difficulty of the problem is the lack of established design standards for qPACT imaging systems, as there is currently a proliferation of qPACT system designs for various applications and it is unknown which ones are optimal or how to best modify the systems under various design constraints. This work introduces a novel computational approach for the optimal experimental design (OED) of qPACT imaging systems based on the Bayesian Cram{\'e}r-Rao bound (CRB). Our approach incorporates several techniques to address challenges associated with forming the bound in the infinite-dimensional function space setting of qPACT, including priors with trace-class covariance operators and the use of the variational adjoint method to compute derivatives of the log-likelihood function needed in the bound computation. The resulting Bayesian CRB based design metric is computationally efficient and independent of the choice of estimator used to solve the inverse problem. The efficacy of the bound in guiding experimental design was demonstrated in a numerical study of qPACT design schemes under a stylized two-dimensional imaging geometry. To the best of our knowledge, this is the first work to propose Bayesian CRB based design for systems governed by PDEs.
\end{abstract}


\vspace{2pc}

\section{Introduction}

Photoacoustic computed tomography (PACT) is a rapidly emerging hybrid medical imaging modality that aims to combine the high contrast of optical imaging modalities with the fine spatial resolution of ultrasound imaging through the use of the photoacoustic effect \cite{wang2012biomedical, tarvainen2024quantitative, poudel2019survey}. In the PACT framework, short optical pulses are used to illuminate the object of interest, leading to spatially varying photon absorption that is converted to heat via thermal relaxation. This causes a local change in pressure, which propagates through the object as an acoustic wave and is detected by ultrasonic transducers located outside the object. These measurements can be used to reconstruct an image corresponding to the induced initial pressure distribution or to estimate the underlying molecular or optical properties of the object. The latter approach is referred to as quantitative PACT (qPACT).

In qPACT, diagnostically useful information such as hemoglobin concentration and blood-oxygen saturation, which have been shown to be important markers of physiological diseases such as cancer, can be imaged \cite{cox2012quantitative}. In particular, qPACT has shown promise in the diagnosis and management of breast cancer, where it can provide valuable information regarding tumor hypoxia and angiogenesis \cite{park2023stochastic, li2015high, lin2018single}.  Another promising application area for qPACT is the diagnosis and treatment of skin diseases. 
Here potential application areas include the diagnosis and management of melanoma, as well as other skin conditions such as dermatitis, psoriasis, and port wine stains \cite{li2021seeing}. There have also been promising developments in the application of qPACT to imaging of the human digestive tract, where it has been demonstrated that qPACT can be used to assess colon function in patients with Crohn's disease \cite{waldner2016multispectral}.

While qPACT is an imaging modality with great promise, the quantitative estimation of molecular concentrations or functional properties using qPACT involves solving a multiphysics inverse problem that is non-linear and severely ill-posed due to the highly scattering nature of biological tissue. The difficulty of the problem is compounded by the fact that, unlike more established imaging modalities (e.g., magnetic resonance imaging \cite{nishimura_2016_principles}), PACT lacks established standards for the design of the imaging system. In fact, for applications such as breast cancer imaging there is currently a proliferation of proposed PACT system designs \cite{lin2018single, toi2017visualization, oraevsky2018full, schoustra2019twente, alshahrani2019all}, and it is unknown which designs are optimal, or how to best modify system designs under constraints (e.g., the ANSI safety limit). Further, optimization of the system design using clinical trials involves ethical considerations and can be expensive. There is therefore a great need for computational methods that can evaluate the quality of existing qPACT system designs and guide their optimization, especially in the early stages of development.

The computational design of qPACT systems is naturally posed as an optimal experimental design (OED) problem \cite{pukelsheim_1993_optimal}. The goal of this paper is to develop a computationally efficient OED approach for qPACT. In particular, in this paper we propose a Bayesian Cram{\'e}r-Rao bound (CRB) based approach for the design of qPACT imaging systems. Originally introduced by Van Trees, the Bayesian CRB provides a lower bound on the expected reconstruction error of any estimator (including both model-based and learned approaches) in a Bayesian inference problem under mild regularity conditions \cite{trees_2007_bayesian, bell_2013_detection}. This property it to serve as the basis for computationally efficient and estimator-independent metrics for guiding system design \cite{Oktel2005, vantrees2006bayesian}. It is also widely used to benchmark the performance of estimators \cite{Robinson2004, Aguerrebere2016}.

The development of Bayesian CRB based approaches for qPACT system design poses several challenges. First, the unknown parameters in qPACT lie in infinite dimensional function spaces. While computational approaches require these parameters to be discretized, it is desirable to have priors that are consistent with the infinite dimensional setting, e.g., pushforwards of Gaussian random fields with \textit{trace-class} covariance operators \cite{stuart2010inverse, bui2013computational}. We address this challenge by using log-Gaussian priors with covariance operators modeled as the squared inverse of a diffusion-reaction partial differential equation (PDE) \cite{lindgren2011explicit}. Second, the Bayesian CRB depends on the (Stein) score of the likelihood function, which in qPACT requires computing derivatives of the PDE-based forward model. Here we use a finite element approach to solve the PDE, which enables derivative computation using the variational adjoint method \cite{troltzsch2010optimal} at the cost of solving one additional PDE, known as the adjoint problem. Finally, qPACT inverse problems generally contain parameters, such as the reduced scattering coefficient, which are unknown but not clinically relevant, complicating the formation of the bound. While there have been several methods proposed in the literature to address such cases \cite{bell_2013_detection, miller78modified}, in this work we adopt the straightforward ``post-marginalization'' approach (see Eq. 4.599 in \cite{bell_2013_detection}) and analyze its performance. 

The computational cost of the proposed approach is dominated by the cost of evaluating the score of the likelihood function, which requires only the evaluation of the PDE-based forward model and the corresponding adjoint PDE. In contrast, many OED approaches (e.g., expected information gain based approaches that require sampling from the posterior \cite{ryan2003estimating}) require solving at least one inverse problem for each design scheme under consideration. The proposed approach is therefore computationally efficient in the sense that, for a given choice of qPACT forward model, evaluation of the Bayesian CRB requires only the evaluation of the forward model and its adjoint and does not require a nested optimization loop with repeated inverse problem solves. As such, the computational cost---measured in terms of forward model evaluations---of the Bayesian CRB approach is relatively low when compared to other OED approaches.

The proposed Bayesian CRB based design approach was investigated in a numerical study of qPACT illumination design schemes. 
Specifically, the case study presented here considers the estimation of the optical absorption coefficient, a key optical imaging parameter, under a stylized two-dimensional model of the imaging geometry. Here the imaging physics were simulated using a multiphysics model that couples a novel computational model for the light illumination subsystem in PACT imaging with a PDE-based model of PACT optical diffusion and a simplified, linear model of PACT acoustic wave propagation. In this setting, we compare the Bayesian CRB to the expected reconstruction error obtained using a maximum a posteriori (MAP) estimator. The results demonstrate that the Bayesian CRB provides an effective surrogate for the expected reconstruction error for a wide range of signal-to-noise ratio (SNR) levels. We then use the Bayesian CRB to rank qPACT design schemes, in two different settings: (1) the setting where the only unknown parameter is the absorption coefficient, and (2) the setting where the reduced scattering coefficient is also unknown and treated as a nuisance parameter. In particular, the qPACT design schemes under consideration  were different multi-illumination angle ``stop-and-go'' schemes \cite{schoustra2019twente} that were constructed based on the design of existing photoacoustic tomography systems \cite{oraevsky2018full, schoustra2019twente}, as well as theoretical considerations regarding the need for multiple illuminations in qPACT \cite{bal2011multi,bal2012multi}.
The results show that in both of the settings described above, the Bayesian CRB correlates well with the expected reconstruction error.

In summary, to the best of the authors' knowledge, this is the first paper to propose Bayesian CRB-based design for optimizing qPACT imaging systems, and in general to perform OED for Bayesian inverse problems governed by PDEs. Further, despite simplifying assumptions in the modeling of the qPACT image formation process, the problem formulation used in this work is still representative of the major challenges of qPACT imaging, including the non-linearity of the imaging operator and the presence of nuisance parameters. The performance of the proposed approach in this setting thus demonstrates its potential in guiding qPACT system design. 

The remainder of this manuscript is organized as follows. In Section \ref{sec:background}, we provide an overview of the optical and acoustic modeling of PACT imaging and introduce the proposed model for the light illumination subsystem. Section \ref{sec:methods} introduces the qPACT inverse problem and proposed Bayesian CRB based design approach for qPACT system optimization. In Section \ref{sec:numerical_studies}, we provide details of the numerical studies that were conducted, while Section \ref{sec:results} discusses the results. Finally, Section \ref{sec:discussion} provides additional comments on the proposed approach, including limitations and future research directions, followed by concluding remarks in Section \ref{sec:conclusion}.

\section{Quantitative PACT Forward Model}
\label{sec:background}

The physics of qPACT can be broken down into two components: an optical component that models light transport within the object of interest, and an acoustic component that models the propagation of the photoacoustically induced ultrasonic pressure wave. In this section, we introduce the optical and acoustic models used in this work. We also propose a novel model for the external light sources used in qPACT systems, and describe how it can be incorporated into the boundary condition of an optical model based on the diffusion approximation (a commonly used approach to model light transport in biological tissue \cite{wang2012biomedical,tarvainen2024quantitative}). Here, and throughout the remainder of the manuscript, we assume that the optical absorption coefficient is the parameter of interest, while the reduced scattering coefficient is also unknown but is treated as a nuisance parameter. 

\subsection{Notation}

We first introduce notation used in this section and throughout the remainder of this manuscript. We use lower-case bold letters to denote vectors (e.g., $\bsx$) and capital bold letters to denote matrices (e.g., $\mathbf{X}$). Scalar-valued functions are denoted using lower case italicized font (e.g., $m$), while vector-valued functions are denoted using lower-case italicized boldface (e.g., $\bm{m}$). Operators between function spaces are denoted using calligraphic font (e.g, $\mathcal{C})$. The gradient of a scalar-valued function $f(\bsx): \mathbb{R}^N \to \mathbb{R}$ is denoted $\nabla_{\bsx} f(\bsx)$ and is written as an $N \times 1$ vector; the subscript in $\nabla_{\bsx}$ is omitted when the differentiation variable is clear from context. The Jacobian of a vector-valued function $F(\bsx): \mathbb{R}^N \to \mathbb{R}^M$ is written as an $M \times N$ matrix, denoted $\nabla_{\bsx} F(\bsx)$, with entries $[\nabla_{\bsx} F(\bsx)]_{i, j} \triangleq \partial F(\bsx)_i / \partial \bsx_j$. The divergence of a vector field $\bm{m}$ is denoted $\nabla \cdot \bm{m}$. 

We use $\mathbf{X}^T$, $\mathbf{X}^{-1}$, and $\| \mathbf{X} \|_2$ to respectively denote the transpose, inverse, and spectral norm of a given matrix $\mathbf{X}$. For two matrices $\mathbf{A}$ and $\mathbf{B}$, the matrix inequality $\mathbf{A} \succcurlyeq \mathbf{B}$ means that $\mathbf{A} - \mathbf{B}$ is a positive semi-definite matrix. The brackets $\langle \cdot, \cdot \rangle$ refer to the standard Euclidean inner product, while $|| \bsx ||_2$ denotes the Euclidean norm of a given vector $\bsx$. For a positive definite matrix $\mathbf{M}$ and vectors $\bsx, \bsy$, we define the weighted norm $|| \bsx ||_{\mathbf{M}}$ as $|| \bsx ||_{\mathbf{M}} \triangleq || \mathbf{M}^{1/2} \bsx ||_2$ and the weighted inner product as $\langle \bsx, \bsy \rangle_{\mathbf{M}} \triangleq \bsx^T \mathbf{M} \bsy$. Given $\mathbf{M}$, the weighted trace of a matrix $\mathbf{X}$ is defined as $\text{tr}_{\mathbf{M}}(\mathbf{X}) = \text{tr} \left ( \mathbf{M} \mathbf{X} \right)$, where $\text{tr}\left( \cdot \right )$ denotes the trace of the given matrix. For a positive definite operator $\mathcal{A}$ and a vector-valued function $\bm{m}$, we define $|| \bm{m} ||_{\mathcal{A}}$ as $|| \bm{m} ||_{\mathcal{A}} \triangleq || \mathcal{A}^{1/2} \bm{m} ||_2$, where for vector-valued functions $\bm{m}$ defined over a domain $\Omega$, $|| \bm{m} ||_2$ denotes the vector $L^2$ norm, i.e., $|| \bm{m} ||_2 \triangleq (\int_{\Omega} || \bm{m}(\bsx) ||_2^2)^{1/2}$. We use $\partial \Omega$ to denote the boundary of a given domain $\Omega$, and $\bseta(\bsx)$ to denote the normal vector of the domain, which is assumed to be well defined for almost every $\bsx \in \partial \Omega$. The expression $\log (\cdot)$ denotes the natural logarithm. 

\subsection{Optical Model} 

In this subsection, we introduce a continuous-to-continuous (C-C) operator $\mathcal{H}_q$ that maps the optical absorption coefficient $\mu_a(\bsx)$ and the reduced scattering coefficient $\mu_s'(\bsx)$ to the absorbed energy density $h(\bsx)$, where $\bsx$ is a point in the imaging domain $\Omega \subset \mathbb{R}^D$. Specifically, we have that 
\begin{equation*}
h = \mathcal{H}_q(\mu_a, \mu_s') \triangleq \mu_a \phi(\mu_a, \mu_s', q),
\end{equation*}
where $\phi(\bsx)$ is the fluence distribution generated by the external light source $q(\bsx)$. The subscript $q$ in the operator $\mathcal{H}_q$ is used to denote the dependency on $q(\bsx)$.

While light is an electromagnetic wave that satisfies Maxwell's equations, the highly scattering nature of biological tissue makes it computationally burdensome to model the fluence distribution as such \cite{cox2012quantitative}. Instead, the diffusion approximation of the radiative transfer equation (RTE) is often used to model the fluence distribution in PACT imaging \cite{ding2015one, bal2011multi, bal2012multi, tarvainen2024quantitative}, as this model is computationally efficient to evaluate and is accurate in the highly scattering regime that characterizes many PACT applications (e.g., breast imaging) \cite{wang2012biomedical}. This is a second-order elliptic PDE that can be written as follows \cite{wang2012biomedical, lee2004modeling}:
\begin{equation}
\left\{
\begin{array}{lr}
     \mu_a(\bsx) \phi(\bsx) - \nabla \cdot D(\bsx) \nabla \phi(\bsx) = 0 & \text{for any} \; \bsx \in \Omega, \\[1mm] 
     \left \langle D(\bsx)  \nabla \phi(\bsx), \bseta (\bsx) \right \rangle + \frac{1}{2}\phi(\bsx) = 2q (\bsx) & \text{for any} \; \bsx \in \partial \Omega.
\end{array}
\right.
\label{eq:diffuse}
\end{equation}
Here, $D(\bsx) \triangleq 1/[3(\mu_a(\bsx) + \mu'_s(\bsx))]$ is the so-called diffusion coefficient, which is a function of the absorption coefficient and the reduced scattering coefficient $\mu'_s(\bsx)$, and $q(\bsx)$ models flux density transmitted into the domain at $\bsx$. Note that for simplicity, the Robin boundary condition above was derived under the assumption that the refractive indices inside and outside of $\Omega$ are matched \cite{lee2004modeling}. As the domain is typically surrounded by a non-scattering medium (e.g., water or air) where the diffusion approximation does not hold, additional modeling is required to capture the light propagation from the sources to the object. In the next subsection, we describe how this formulation can be coupled with a model of acoustic wave propagation. We then introduce the model for the light sources. 

\subsection{Acoustic Model}

In PACT, the absorbed optical energy is converted to a local increase in pressure via the photoacoustic effect, which then propagates through the object as an acoustic wave. Mathematically, this can be expressed as $p_0(\bsx) = \Gamma(\bsx) h(\bsx)$, where $p_0(\bsx)$ is the initial pressure and $\Gamma(\bsx)$ is the Gr{\"u}neisen parameter, which measures the photoacoustic efficiency of the tissue. In this work, the following simplifying assumptions are used: 1) the Gr{\"u}neisen parameter is known and constant in space (we assume $\Gamma(\bsx) \equiv 1$ for simplicity), and 2) the medium is a lossless and acoustically homogeneous medium. The validity and impact of these modeling choices is further discussed in Section \ref{sec:discussion} .

Under the above assumptions, the solution to the wave equation can be expressed in closed form \cite{xu2006photoacoustic}:
\begin{equation}
\label{eq:wave_eq_sol}
p(\bsx', t) = \frac{1}{4\pi c_0^2} \int_{\Omega} h(\bsx) \frac{d}{dt} \frac{\delta\left(t - \frac{||\bsx' - \bsx||_2}{c_0}\right)}{||\bsx' - \bsx||_2} \; d\bsx,
\end{equation}
where $c_0 \in \mathbb{R}$ is the sound speed in the medium and $\delta(\cdot)$ is the Dirac delta distribution. The above expression can be conveniently written in terms of a spherical Radon transform (SRT) \cite{poudel2019survey}:
\begin{equation}
d(\bsx', t) =  \int_{\Omega}  h(\bsx) \delta\left(c_0 t - ||\bsx' - \bsx||_2\right) \; d\bsx,
\label{eq:SRT}
\end{equation}
where the data function $d(\bsx', t)$ is related to $p(\bsx', t)$ by 
$$
p(\bsx', t) =  \frac{1}{4\pi c_0^2} \frac{\partial}{\partial t} \left ( \frac{d(\bsx', t)}{t}\right ).
$$
In the case where the domain is two-dimensional, the SRT reduces to a circular Radon transform (CRT). 

We assume the induced wave is measured by a set of idealized point-like ultrasonic transducers located outside of the domain of interest at discrete time points. We can therefore define a continuous-to-discrete (C-D) mapping that maps the absorbed energy $h(\bsx)$ to the temporally sampled measurements $\bsd \in \mathbb{R}^K$, where $K$ is the product of the number of temporal samples and ultrasonic transducers. Under the assumption of additive Gaussian noise, this can be written as 
\begin{equation}
    \bsd = \mathcal{H}_{a} h + \bsz, \quad \bsz \sim \mathcal{N}(\mathbf{0}, \boldsymbol{\Sigma}_{\mathbf{z}}).
\end{equation}
Here  $\boldsymbol{\Sigma}_{\mathbf{z}} \in \mathbb{R}^{K \times K}$ is the covariance matrix of the noise and $\mathcal{H}_{a}$ is the (linear) continuous-to-discrete (C-D) acoustic operator corresponding to  evaluation of $d(\bsx', t)$ at discrete spatial and temporal points.  

\subsection{Illumination Subsystem}
In this work, we consider qPACT experiments with illumination subsystems consisting of a set of cone beam light sources, which we assume are located in a region outside the domain $\Omega$ that is absorption-dominated with negligible scattering effects. Under these assumptions, the flux density at $\bsx \in \partial \Omega$ associated with a source with power $P$ located at position $\bss \in \mathbb{R}^D$ can be written as 
$$
\bm{q}(\bsx; \bss) = P  g(\theta)  \frac{e^{- \mu_a' || \bss - \bsx||_2}}{4\pi ||\bss - \bsx||_2^2} \bsu.
$$
Here $\bsu \triangleq (\bss - \bsx)/||\bss - \bsx||_2$ is the ray direction, $\mu_a' \in \mathbb{R}$ is the optical absorption coefficient outside the domain, which we assume is constant, and $g(\theta)$ defines the modulated intensity of the beam with respect to the angle $\theta$ between the direction of the cone-beam source and the ray $\bsu$. Note that the above expression is derived from the modified monopole solution for isotropic light sources (see, e.g., \cite{pharr2023physically}) under the assumption that the power of the cone beam decays with $\theta$; see Figure \ref{fig:illumination_fig} for an illustration. The total flux density transmitted into the domain at the point $\bsx \in \partial \Omega$ can therefore be written as
$$
q(\bsx) = \sum_{i=1}^S \mathrm{max} \left ( \left \langle \bm{q}(\bsx; \bss_i), \bseta(\bsx) \right \rangle, 0 \right ),
$$
where $\{\bss_i \}_{i=1}^S$ is the set of source locations and $\boldsymbol{\eta}(\bsx)$ is the unit normal vector to the boundary of $\Omega$ at $\bsx$. This defines the boundary term in \eqref{eq:diffuse}. 
\begin{figure}
    \centering
    \includegraphics[width = .4\textwidth]{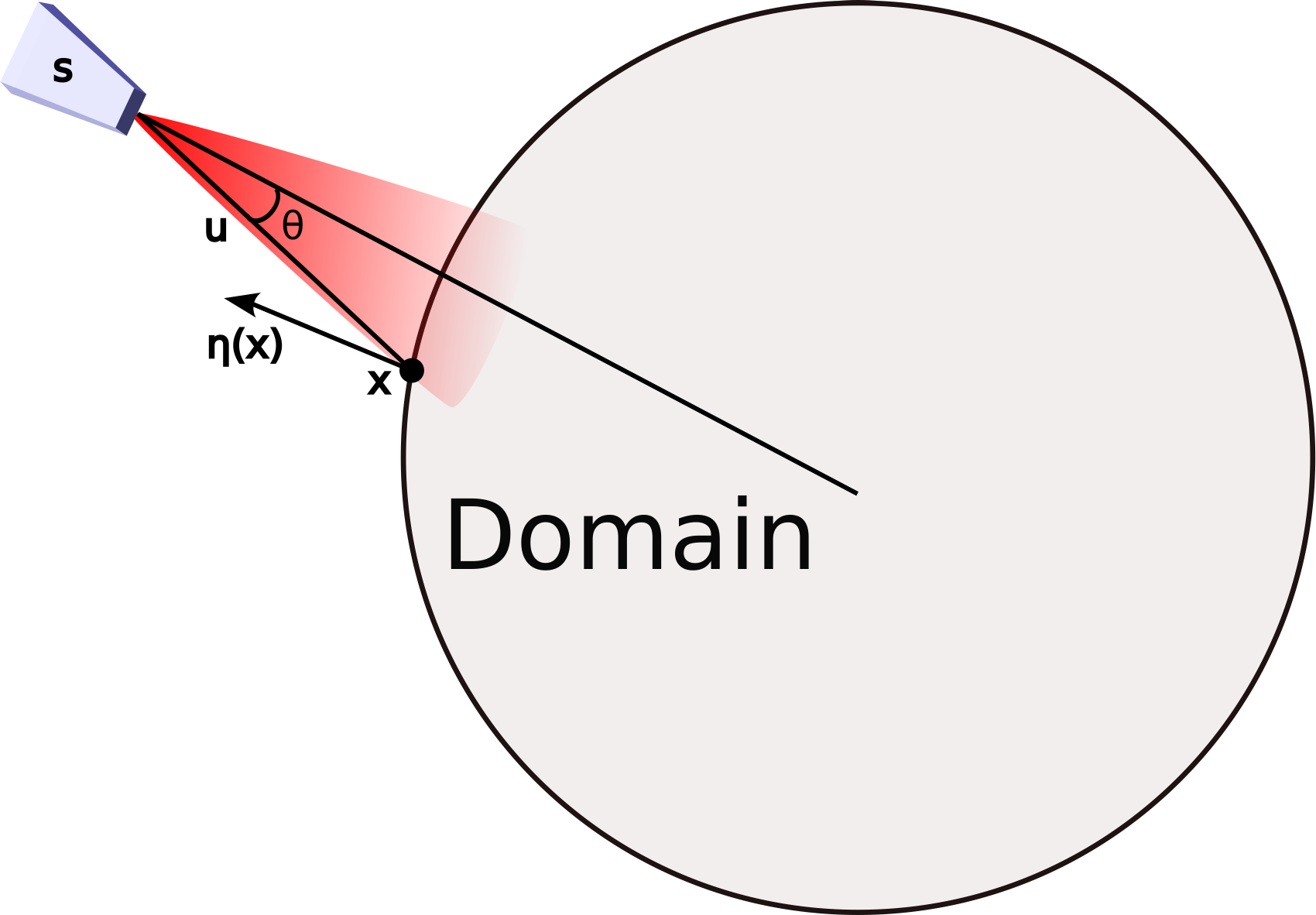} \\
    \caption{Illustration of the cone-beam illumination system model. The power of the cone-beam is non-isotropic and decays as the angle $\theta$ between the direction of the cone-beam and $\bsu$ increases, where $\bsx$ is a point on the domain's boundary and $\bss$ is the location of the source. The flux density transmitted into the domain at $\bsx$ then depends on angle between $\bsu$ and the domain normal vector $ \bseta(\bsx)$ at $\bsx$.} 
    
    \label{fig:illumination_fig}
\end{figure}

\section{Inverse Problem and OED Approach}
\label{sec:methods}

This section introduces the proposed Bayesian CRB based approach for qPACT OED. Here we first formulate the qPACT reconstruction problem as a Bayesian inverse problem, in which the optical absorption coefficient and reduced scattering coefficient are parameterized as functions of a latent inversion parameter $\bm{m}$, and motivate the proposed Bayesian CRB based OED approach. We then introduce the Bayesian CRB and the proposed approach for computing the bound for $\bm{m}$ under the assumption that the regularity conditions required to form the bound hold. Finally, we describe how the bound on $\bm{m}$ can be transformed into a bound on the optical absorption coefficient and used as the basis for design metrics in qPACT OED. 

\subsection{Problem Formulation}

We first formulate the qPACT reconstruction problem as a Bayesian inference problem. The Bayesian setting enables us to capture a key feature of the OED problem: A good experimental design scheme should perform well on any physiologically plausible realization of the unknown parameters $\mu_a$ and $\mu_s'$, not just a small subset of representative examples. Here it is crucial to choose a prior distribution that captures the statistics of the unknown parameters, while also ensuring that the corresponding statistical inverse problem is well-posed. To this end, we propose parameterizing $\mu_a$ and $\mu_s'$ as $\mu_a(\bsx) = \overline{\mu_a} e^{m_1(\bsx)}$ and $\mu_s'(\bsx) = \overline{\mu_s'} e^{m_2(\bsx)}$. Here $\overline{\mu_a}$ and $\overline{\mu_s'}$ are reference values for the optical absorption coefficient and reduced scattering coefficient, respectively, while $\bm{m}(\bsx) \triangleq [m_1(\bsx), m_2(\bsx)]^T$ is a zero-mean Gaussian random vector field, i.e., $\bm{m}\sim \mathcal{N}(\bm{0}, \mathcal{C})$, that represents the logarithms of the absorption and reduced scattering coefficients. Here $\mathcal{C}$ is the covariance operator, which for simplicity we assume has block diagonal structure. The block components $\mathcal{C}_1$ and $\mathcal{C}_2$ corresponding to $\mu_a$ and $\mu_s'$, respectively, are differential operators that can be written as $\mathcal{C}_1 = \mathcal{A}_1^{-2}$ and $\mathcal{C}_2 = \mathcal{A}_2^{-2}$, where $\mathcal{A}_1$ and $\mathcal{A}_2$ are diffusion-reaction PDE operators \cite{bui2013computational}. 

The choice of prior introduced above has several key benefits. First, the exponential parameterization preserves the nonnegativity of $\mu_a$ and $\mu_s'$. Second, the covariance operator is trace-class \cite{stuart2010inverse}, which ensures that the prior distribution has bounded variance. Third, the differential structure of the covariance operator leads to prior samples that have spatial structure and can therefore capture the general structure of patterned occlusions found in biological tissue. Finally, the covariance operators have simple and easy-to-apply square roots, i.e., $\mathcal{C}_1^{1/2} = \mathcal{A}_1^{-1}$ and $\mathcal{C}_2^{1/2} = \mathcal{A}_2^{-1}$, which yields fast sampling from the prior distribution \cite{villa2021hippylib}. 

In qPACT experiments, multiple illuminations of the object of interest are usually used to mitigate the ill-posedness of the inverse problem \cite{bal2011multi}. Let $\bsy \triangleq \{ \bsd_i \}_{i=1}^I$, where $I$ is the total number of illuminations and $\bsd_i$ is the data obtained from the $i$-th illumination. Then under the above assumptions, by Bayes' theorem we have that the posterior measure $\mu^y$ of $\bm{m}$ is related to the prior measure $\mu_0$ by
\begin{equation}
\label{eq:posterior}
\frac{d \mu^y}{d \mu_0} = p (\bsy \mid \bm{m}) \propto  \prod_{i=1}^I \mathrm{exp} \left ( - \frac{1}{2} || \mathcal{H}_{a} h_i - \bsd_i ||_{\boldsymbol{\Sigma}_\mathbf{z}^{-1}}^2 \right), \quad h_i = \tilde{\mathcal{H}}_{q_i}(\bm{m}),
\end{equation}
where $\tilde{\mathcal{H}}_{q_i}(\bm{m}) \triangleq \mathcal{H}_{q_i}(\overline{\mu_a} e^{m_1}, \overline{\mu_s'} e^{m_2})$ and $q_i(\bsx)$ is the $i$-th illumination. Here, \eqref{eq:posterior} is understood as the Radon-Nikodym derivative of $\mu^y$ with respect to $\mu_0$ \cite{stuart2010inverse}.

Different point estimators of the unknown parameters can be constructed from the posterior measure. For example, the MAP estimate of the unknown parameters can be obtained be solving the following optimization problem \cite{stuart2010inverse}:
\begin{gather}
\argmin_{\bm{m}} \;  \frac{1}{2} \sum_{i=1}^I ||\mathcal{H}_{a} h_i - \bsd_i ||_{\boldsymbol{\Sigma}_{\mathbf{z}}^{-1}}^2 + \frac{1}{2} || \bm{m} ||_{\mathcal{C}^{-1}}^2 \nonumber \\
\text{subject to } h_i = \tilde{\mathcal{H}}_{q_i}(\bm{m}), \quad i =1, \dots, I. \label{eq:map_problem}
\end{gather}
This is a non-linear and non-convex optimization problem that can be discretized using finite element approaches \cite{brenner2008mathematical}.

The goal of this work is to develop an OED approach for the above inverse problem. Specifically, in this work we focus on the optimization of the illumination subsystem. A straightforward approach to this problem would be to choose a point estimator for the qPACT inverse problem (e.g., a MAP estimator), and select the design scheme that minimizes the expected reconstruction error under this estimator. However, there are two main reasons why this approach may not be desirable in the guidance of imaging system design. First, the optimized design scheme would be dependent on the specific choice of estimator, and may not perform well with other estimators. For example, a design that performs well when using a MAP estimator may perform differently when using a minimum mean square error (MMSE) estimator or a learning-based estimator. Further, as the MAP estimation problem is non-convex and challenging to solve, even the choice of optimization algorithm used to solve the MAP problem could impact the optimized design scheme. Second, it is computationally expensive to compute the expected reconstruction error under a given estimator and design scheme, as this would require repeatedly solving the qPACT inverse problem for a prohibitively large number of samples of $\bm{m}$ and $\bsy$ until convergence is obtained. To address these issues, we propose a novel approach to qPACT OED based on the Bayesian CRB, which we now introduce. 

\subsection{Forming the Bound}

The Bayesian CRB provides an estimation-theoretic lower bound on the error of \emph{any} estimator in a Bayesian inverse problem. Specifically, let $\bsm \in \mathbb{R}^N$ be a discretization of $\bm{m}$ in \eqref{eq:map_problem}, where $\bsm = [\bsm_1^T \; \bsm_2^T]^T$ and $\bsm_1 \in \mathbb{R}^{N_1}$ and $\bsm_2 \in \mathbb{R}^{N_2}$ ($N_1 + N_2 = N$) are discretizations of $m_1(\bsx)$ and $m_2(\bsx)$, respectively. Further, let $p(\bsm, \bsy)$ denote the joint probability density  $p(\bsm, \bsy)$ of $\bsm$ and $\bsy$, defined as 
\begin{equation}
\label{eq:joint_probability}
p(\bsm, \bsy) \triangleq \exp\left\{ - \frac{1}{2} \| \bsg(\bsm) - \bsy \|^2_{\boldsymbol{\Sigma}_\bsz^{-1}} \right\} \exp\left\{ - \frac{1}{2} \| \bsm \|^2_{\bsC^{-1}} \right\},
\end{equation}
where $\bsg(\bsm)$ stems from the discrete counterpart of the continuous-to-discrete qPACT imaging operator $\{\mathcal{H}_a\tilde{\mathcal{H}}_{q_i}(\bm{m})\}_{i=1}^I$ and $\bsC$ is the covariance matrix of $\bsm$ and is obtained by discretizing $\mathcal{C}$. Then under mild regularity conditions (see \ref{sec:appendix_wellposed})  the expected error of an estimator $\hat{\bsm}(\bsy)$ of $\bsm$ can be described by the following information inequality \cite{bell_2013_detection}:
\begin{equation}
    \mathbb{E}_{\bsm, \bsy} \left [\left (\hat{\bsm}(\bsy) - \bsm \right )(\hat{\bsm} \left (\bsy) - \bsm \right )^T \right ] \succcurlyeq \bsV_{\bsm } \triangleq \bsJ_{\bsm }^{-1}. \label{bound}
\end{equation}
Here the expectation is computed over the joint distribution $p(\bsm, \bsy)$, while $\bsV_{ \bsm } \in \mathbb{R}^{N \times N}$  and $\bsJ_{\bsm }  \in \mathbb{R}^{N \times N}$ are the Bayesian CRB and Bayesian information, respectively, for the parameter $\bsm$, with the Bayesian information defined as 
\begin{equation}
    \bsJ_{\bsm } \triangleq \mathbb{E}_{\bsm, \bsy} \left [  \nabla_\bsm \log p(\bsm, \bsy) \; \nabla_\bsm \log p(\bsm, \bsy)^T \right ].
    \label{eq:jb_definition}
\end{equation}
The above information inequality enables the Bayesian CRB to serve as a computationally efficient and estimator-independent surrogate for the expected error covariance in Bayesian inverse problems, which facilitates its use as the basis for OED metrics \cite{chaloner1995bayesian}. 

To compute the bound, we use the following decomposition of the Bayesian information \cite{bell_2013_detection}: 
\begin{equation}
\bsJ_{\bsm } = \bsJ_P + \bsJ_D.
\label{eq:Jb_exact}
\end{equation}
Here, $\bsJ_P$ is known as the prior term and is given by  
\begin{equation}
\bsJ_P \triangleq \mathbb{E}_{\bsm} \left [ \nabla_{\bsm} \log p(\bsm) \; \nabla_{\bsm} \log p(\bsm)^T \right ] =  - \mathbb{E}_{\bsm} \left [  \nabla_{\bsm} [\nabla_{\bsm} \log p(\bsm) ] \right ],
\label{eq:Jp}
\end{equation}
where  the second equality can be obtained via integration by parts (see, e.g., \cite{bell_2013_detection}). Further, the data term $\bsJ_D$ can be written as
\begin{equation}
\bsJ_D \triangleq \mathbb{E}_{\bsm, \bsy} \left [ \nabla_\bsm \log p(\bsy \, | \, \bsm) \; \nabla_\bsm \log p(\bsy \, | \, \bsm)^T\right ], 
\label{eq:Jd}
\end{equation} 
In our setting, $\bsm$ is Gaussian distributed, and as a result $\bsJ_P$ can be straightforwardly computed in closed form, i.e., $\bsJ_P = \bsC^{-1}$, where $\bsC$ is the covariance matrix of $\bsm$. 
To compute $\bsJ_D$, we use Monte-Carlo sampling from the joint distribution $p(\bsm, \bsy)$ to approximate the expectation in \eqref{eq:Jd}. Specifically, we form the following estimate $\hat{\bsJ}_D$ of $\bsJ_D$:
\begin{equation}
\label{eq:Jd_estimator}
\hat{\bsJ}_D \triangleq \frac{1}{N_s} \sum_{i=1}^{N_s}  \nabla_\bsm \log p(\bsy_i \, | \, \bsm_i) \; \nabla_\bsm \log p(\bsy_i \, | \, \bsm_i)^T.
\end{equation}
Here $\{\bsm_i, \bsy_i\}_{i=1}^{N_s}$ are i.i.d. samples from $p(\bsm, \bsy)$, and $\nabla_\bsm \log p(\bsy_i \, | \, \bsm_i)$, the (Stein) score of the likelihood, is obtained using the variational adjoint method at the cost of solving on additional PDE known as the adjoint problem; see \ref{sec:appendix_deriv} for details. We then form the following estimates $\hat{\bsJ}_{\bsm}$ and $\hat{\bsV}_{\bsm}$ of the Bayesian information and Bayesian CRB for $\bsm$, respectively:
$$
\hat{\bsJ}_{\bsm} \triangleq \bsC^{-1} + \hat{\bsJ}_D, \quad \hat{\bsV}_{\bsm} \triangleq \hat{\bsJ}_{\bsm}^{-1}.
$$
Note that $\hat{\bsJ}_{\bsm}$ is the sum of a symmetric positive definite and a positive semi-definite matrix and is therefore invertible, which ensures that $\hat{\bsV}_{\bsm}$ is well defined.

\subsection{Bound Transformation and Design Metric}

As previously noted, the optical absorption coefficient is the parameter of interest in our problem setting, while the reduced scattering coefficient is treated as a nuisance parameter. Several approaches have been introduced in the literature to form the Bayesian CRB in the presence of nuisance parameters \cite{bell_2013_detection, miller78modified}. Here we consider a straightforward approach: The bound on $\bsm_1$ is obtained by simply taking the block of $\bsV_{\bsm}$ corresponding to the optical absorption coefficient (see, e.g., Eq. 4.599 in \cite{bell_2013_detection}). We refer to this strategy as ``post-marginalization,'' because the impact of the nuisance parameter on the estimation problem is marginalized out after forming the bound. We denote the estimate of $\bsV_{\bsm_1}$ obtained from $\hat{\bsV}_{ \bsm }$ using this approach as $\hat{\bsV}_{ \bsm_1 }$.

The bound on $\bsm_1$ can be transformed into a bound on the reconstruction error of the optical absorption coefficient using the change-of-variable formula for the Bayesian CRB \cite{bell_2013_detection}. The change-of-variable formula enables the computation of a reconstruction bound on any parameter $\bsr \in \mathbb{R}^{M}$ that can be written as $\bsr \triangleq f(\bsm_1)$, where $f: \mathbb{R}^{N_1} \to \mathbb{R}^{M}$ is an arbitrary differentiable function. For estimators $\hat{\bsr}(\bsy)$ of $\bsr$, we then have that \cite{bell_2013_detection}:
\begin{equation}
\label{eq:change_variable}
\mathbb{E}_{\bsm, \bsy} \left [\left (\hat{\bsr}(\bsy) - f(\bsm_1) \right )(\hat{\bsr} \left (\bsy) - f(\bsm_1) \right )^T \right ] \succcurlyeq \bsC^T \; \bsV_{\bsm_1} \; \bsC, \quad \bsC \triangleq \mathbb{E}_{\bsm} \left [\nabla_{\bsm_1} \left[f(\bsm_1) \right] \right ].
\end{equation}
In our setting, we define $\bsr \triangleq \boldsymbol{\mu}_a = \overline{\mu_a} e^{\bsm_1}$, where $\boldsymbol{\mu}_a \in \mathbb{R}^{N_1}$ is a discretization of $\mu_a(\bsx)$, and $f$ takes the form of a pointwise exponential transformation. This yields the following estimator for the bound on the estimation of $\boldsymbol{\mu}_a$:
\begin{equation}
\label{eq:change_variable_estimate}
\hat{\bsV}_{ \boldsymbol{\mu}_a } \triangleq \hat{\mathbf{C}}^T \; \hat{\bsV}_{\bsm_1} \; \hat{\mathbf{C}},
\end{equation}
where $\hat{\mathbf{C}}$ is an estimator of $\mathbf{C}$ obtained using Monte-Carlo sampling.\footnote{In this work, the prior distribution is parameterized using an exponential tranformation, and a closed-form expression for $\mathbf{C}$ can be obtained from the moment generation function of the multivariate Gaussian distribution. However, here Monte-Carlo sampling is used for generality.}

In our approach, $\hat{\bsV}_{ \boldsymbol{\mu}_a }$ serves as the basis for evaluating the quality of a given design scheme.  Specifically, in this work, we aim to choose design schemes that minimize the mean square error (MSE) in estimates of $\mu_a$. This corresponds to minimizing $\text{tr}_{\mathbf{M}}\left( \hat{\bsV}_{ \boldsymbol{\mu}_a } \right)$, where $\mathbf{M} \in \mathbb{R}^{N_1 \times N_1}$ is the mass matrix corresponding to the finite element discretization of $\mu_a$ and is needed to make the design criterion consistent with the infinite-dimensional setting \cite{bui2013computational, alexanderian2014optimal}. This design metric can be viewed as the Bayesian analogue of the popular CRB-based A-optimality design criterion \cite{pukelsheim_1993_optimal}.

\section{Numerical Studies}
\label{sec:numerical_studies}

The efficacy of the proposed approach in providing a computationally efficient and estimator-independent metric for qPACT system design was evaluated in a set of numerical studies. In the first study, the proposed Bayesian CRB estimator was validated by comparing the estimated bound with the reconstruction error at different SNR levels in a simplified setting with a single-illumination design scheme and known scattering coefficient. In the second study, we used the bound to rank the performance of two different multi-illumination qPACT design schemes in the setting where the reduced scattering coefficient is assumed to be known. The third study extends the second study to the unknown scattering coefficient setting. The related details of the three numerical studies are provided in the following subsections.

\subsection{Problem Setting}

We first describe the study details common to all three numerical studies, including the imaging geometry and prior distribution. In this work, we used a stylized two-dimensional virtual imaging system with a circular acoustic detection geometry. This choice was made for the sake of simplicity but can also be motivated by the fact that, while qPACT is a three-dimensional imaging technology, some qPACT systems use elevationally focused illuminators and ultrasonic transducers, which exhibit quasi two-dimensional geometry \cite{alshahrani2019all}. 

The acoustic detection geometry consisted of $360$ ultrasonic transducers uniformly placed in a circle of radius $6 \; \rm{cm}$ centered around the domain, i.e., the transducers were placed $1 \; \rm{cm}$ from the domain boundary. The speed of sound in the domain was set as $c_0 = 1500 \; \rm{m/s}$. Each transducer recorded measurements at $184$ time points with sampling interval $\approx 2 \times 10^{-7} \; \rm{s}$ for a total data dimension of $K = 66240$. The covariance matrix $\boldsymbol{\Sigma}_{\bsz}$ for the measurement noise was set as $\boldsymbol{\Sigma}_{\bsz} = \sigma^2 \mathbf{I}$, where $\mathbf{I}$ is the identity matrix and $\sigma^2 = 10^{-3}$ except where otherwise noted. 

In all studies, the object support $\Omega \subset \mathbb{R}^2$ was a circle of radius $5 \; \mathrm{cm}$ located at the center of  of the imaging system.
The optical properties were modeled following a log-Gaussian prior distribution on $\mu_a$ with reference value $\overline{\mu_a} = e^{-2} \approx 0.1 \; {\rm cm}^{-1}$. The corresponding covariance operator $\mathcal{C}_1$ was chosen such that the pointwise prior variance was uniformly $0.2$ and the correlation length, i.e., the distance between two points $\bsx_1, \bsx_2 \in \Omega$ such that their correlation coefficient is $0.1$, was $5 \; \mathrm{cm}$. In Studies 1 and 2, where the reduced scattering coefficient is fixed, we set $\mu_s'(\bsx) \equiv 10  \; {\rm cm}^{-1}$. In Study 3, the prior on $\mu_s'$ was implemented with reference value $\overline{\mu_s'} = 10  \; {\rm cm}^{-1}$ and the covariance operator $\mathcal{C}_2$ was implemented with pointwise variance $0.05$ and correlation length $5 \; \mathrm{cm}$. The top row of Figure \ref{fig:mua_recon_val} shows four i.i.d. samples of $\mu_a$ from the prior. As can be seen, the samples capture the general structure of patterned occlusions in biological tissues and the range of $\mu_a$ values is plausible.



\subsection{Design Schemes}

For validation in the simplified single-illumination setting of Study 1, we fixed the design scheme so that $q(\bsx) \equiv 1 \; \rm{AU}$, where here $1 \; \rm{AU}$ represents the ANSI safety limit on light exposure. In Studies 2 and 3, the multi-illumination design schemes both use $S = 10$ cone-beam sources located on a circle of radius $10 \; \mathrm{cm}$ centered around the domain. A total of $I = 4$ illuminations were used for each multi-illumination design scheme. In the first scheme, the light sources were placed uniformly along an arc of the circle subtended by an angle of $90 ^{\circ}$; the light sources were then rotated by $\alpha = \pi/2$ radians after each of the four illuminations to ensure complete coverage of the domain. We refer to this scheme as the contiguous design scheme. In the second design scheme, the ten light sources were placed uniformly around the circle, and, after each illumination, were rotated by an angle $\alpha = \pi/20$ radians. We refer to this scheme as the interlaced design scheme.  An illustration of the two schemes is given in Figure \ref{fig:design_scheme}. 

\begin{figure}
    \centering
    \includegraphics[width = \textwidth]{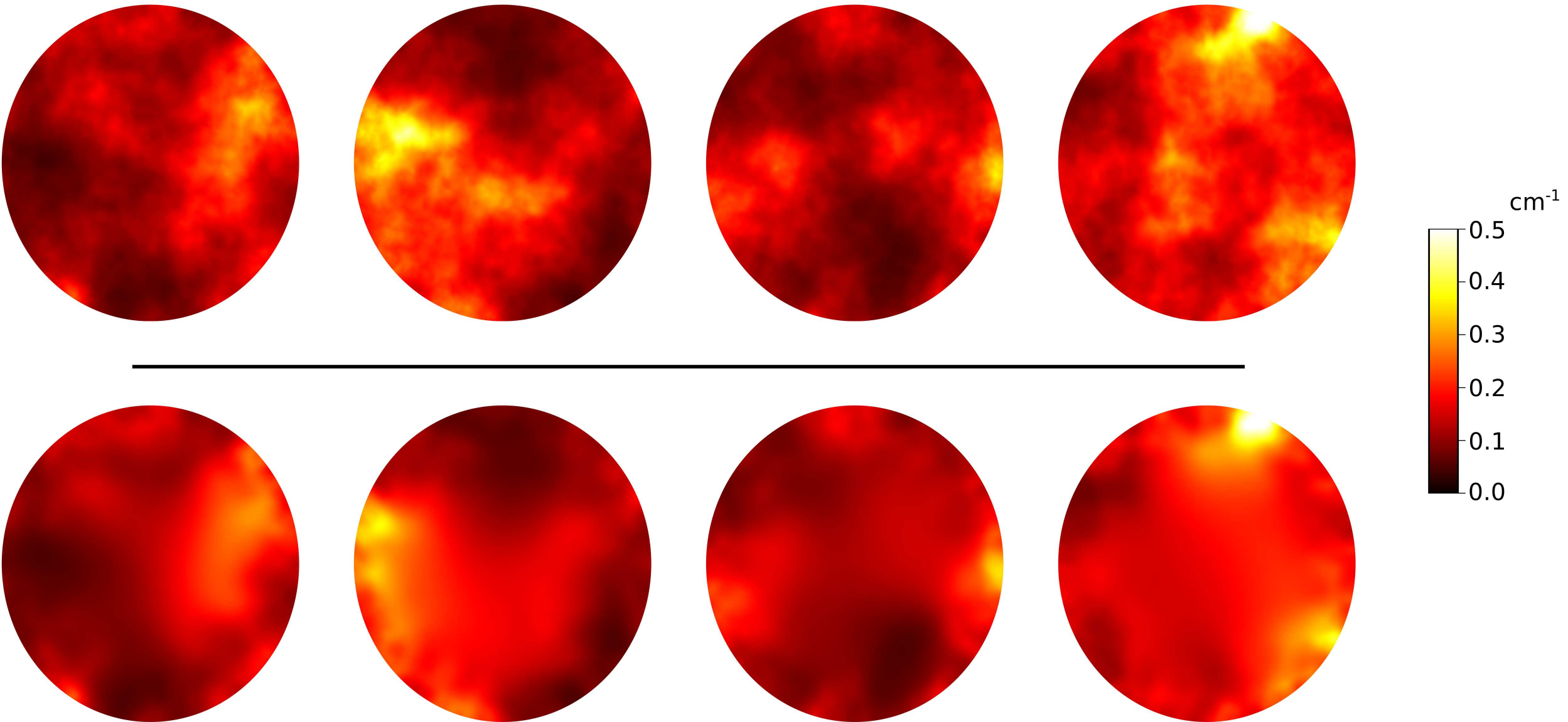}
    \caption{Study 1: Four i.i.d. samples of the absorption coefficient $\mu_a$ from the log-Gaussian prior (top) and the corresponding reconstruction results in the single-illumination setting (bottom). As can be seen, the samples have the general structure of patterned occlusions found in biological tissue and the range of $\mu_a$ values is physiologically plausible. The reconstructions capture the main features of the prior samples, with the best performance near the boundary of domain, where the signal-to-noise ratio is highest.}
    \label{fig:mua_recon_val}
\end{figure}

For each design scheme, the power $P$ of the sources was normalized so that fluence had a maximum value equal to the ANSI safety limit ($1 \; \rm{AU}$). The modulation function of the sources was set as $g(\theta) = \cos(\theta) \, \1(| \theta | < \beta/2)$, where $\1$ is the indicator function, $\beta = 25^{\circ}$ is the aperture angle of the sources, and the sources all point towards the center of the imaging system. The absorption coefficient for the surrounding medium was set as $\mu_a' = 10^{-3} \; \rm{cm}^{-1}$ .  Figure \ref{fig:fluence_fig} shows the fluence distributions from a single illumination and the superimposed fluences from the four illuminations for the interlaced and contiguous design schemes, where here the absorption coefficient and reduced scattering coefficient were set equal to their prior means. As can be seen, both design schemes illuminate the entire domain, with the interlaced design scheme providing more fluence overall. 

\begin{figure}
    \centering
    \hspace{1cm} \includegraphics[width = .95\textwidth]{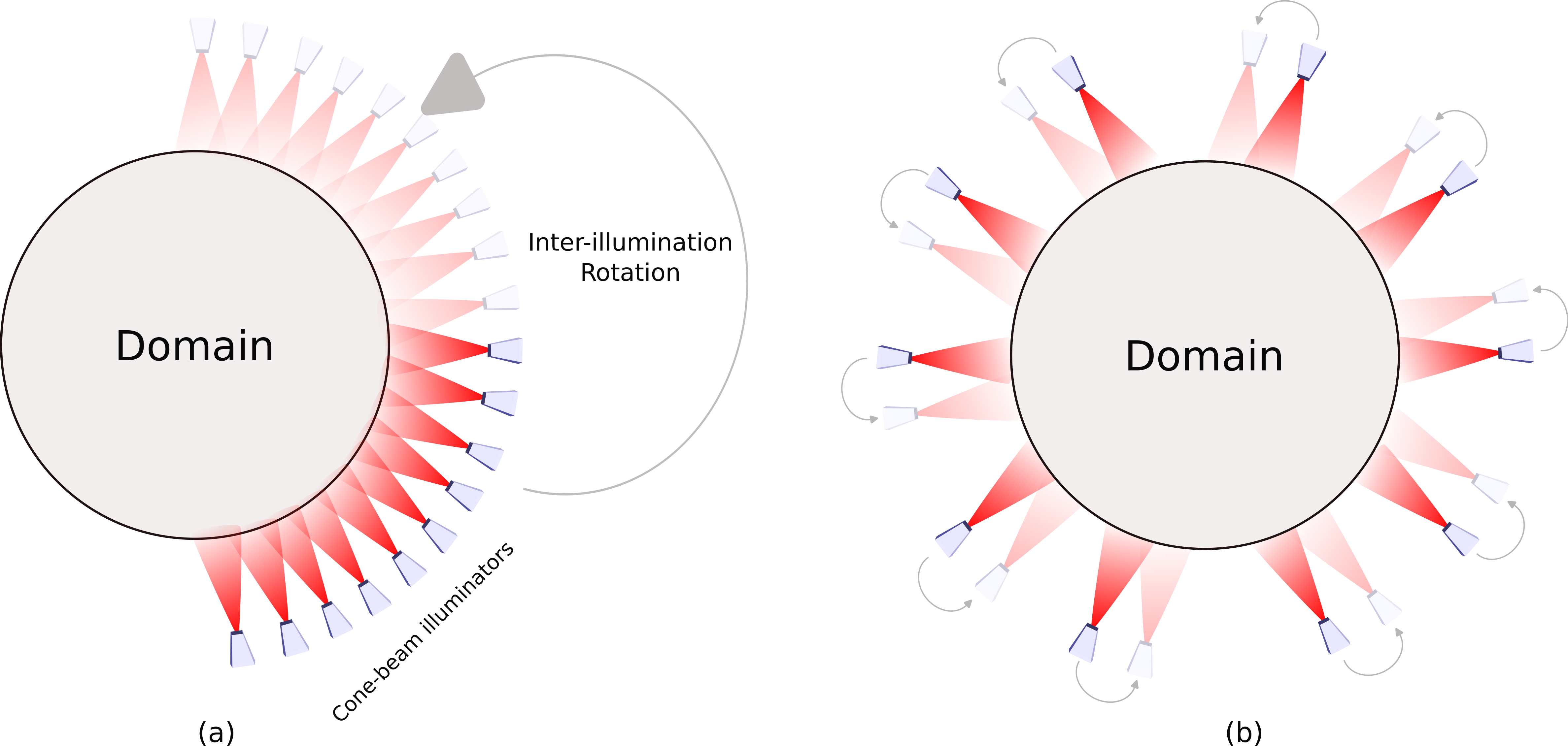}
    \caption{Illustration of the (a) contiguous and (b) interlaced design schemes. In each design scheme, cone-beam illuminators are placed in a circle around the domain. After each illumination, the illuminators are rotated by an angle $\alpha$ along the circle, where $\alpha = \pi/2$ for the contiguous scheme and $\alpha = \pi/20$ for the interlaced scheme.}
    \label{fig:design_scheme}
\end{figure}

\subsection{Implementation of Proposed Approach}

The optical and acoustic models were implemented in hIPPYlib \cite{villa2018hippylib, villa2021hippylib}, a scalable Python library for PDE-based inverse problems. Here the unknown parameter $\bm{m}$ and all other parameters defined over the domain $\Omega$ were discretized using the finite-element approach with first-order continuous Galerkin elements. The finite element mesh was unstructured with triangular elements, and the element diameter set to smoothly interpolate between $0.08 \; \rm{cm}$ near the boundary and $.15 \; \rm{cm}$ in the center of the domain.  The resulting parameter dimension was $N_1 = N_2 = 7499$, i.e., $N = 14998$.  The acoustic model was implemented using a discretization of the CRT in \eqref{eq:SRT} that extends the CRT implementation in \cite{hansen2018air} to unstructured grids and finite element representations of the object function. Specifically, for each sensor location, the finite element representation of $h(\bsx)$ was integrated along the concentric arcs corresponding to the sampling time points. 

In all numerical studies, the proposed approach was implemented with $N_s = 5000$ samples used to estimate $\bsJ_D$; $5000$ samples were also used to estimate the change-of-variable matrix $\hat{\mathbf{C}}$ in \eqref{eq:change_variable_estimate}. All bound computation was done on a Linux server with two Intel Xeon Gold 5218 processors (with $2.30$ GHz and $32$ cores each) and 386 GB RAM running Python 3.8.16.

The proposed approach was validated through comparison to the expected reconstruction error. Here the estimator $\hat{\bsm}(\bsy)$ was implemented as a MAP estimator using the inexact Newton-CG algorithm with early termination of the CG iterations via Eisenstat-Walker (to prevent oversolving) and Steihaug (to avoid negative curvature) criteria \cite{villa2021hippylib, nocedal1999numerical}. Globalization was achieved using line search. The solver was implemented on the Frontera supercomputer at the Texas Advanced Computing Center (TACC), with each inverse problem solve running on a single core of an Intel Xeon Platinum 8280 processor running Python 3.6.9. In each experiment, the expectation in \eqref{bound} was approximated using $10000$ inverse problem solves unless otherwise noted.

\section{Results}
\label{sec:results}

In this section, representative results are shown to illustrate the performance of the proposed approach. 

\subsection{Study 1: Bound Validation}

We first validated the proposed Bayesian CRB-based design approach in the single-illumination setting with known scattering coefficient. Figure \ref{fig:mua_recon_val} shows i.i.d. samples from the prior distribution on $\mu_a$ and corresponding reconstructions in this setting. Despite the ill-posedness of the inverse problem, the reconstruction algorithm is able to capture the main features of the ground-truth absorption coefficient maps, especially near the boundary of the domain, where the SNR is highest.


Figure \ref{fig:recon_vs_bound_val} shows the pointwise MSEs and the corresponding Bayesian CRB bounds for the latent parameter $m_1$ and $\mu_a$. As can be seen, the proposed approach provides a tight lower bound on the error in $m_1$. The change-of-variable transformation used to compute the $\mu_a$ bound causes a loss of bound sharpness. However, the $\mu_a$ bound is still fairly tight, captures the general structure of the expected error, and, as we will show, provides an effective surrogate for the reconstruction error in experimental design. Further, as can be seen in Figure \ref{fig:snr_figure}, which plots the MSE and corresponding Bayesian CRB bound as a function of SNR, these observations hold for a wide range of SNR levels, with no significant loss in sharpness in either the low or high SNR regimes. 

\begin{figure}
    \centering
    \includegraphics[width=\textwidth]{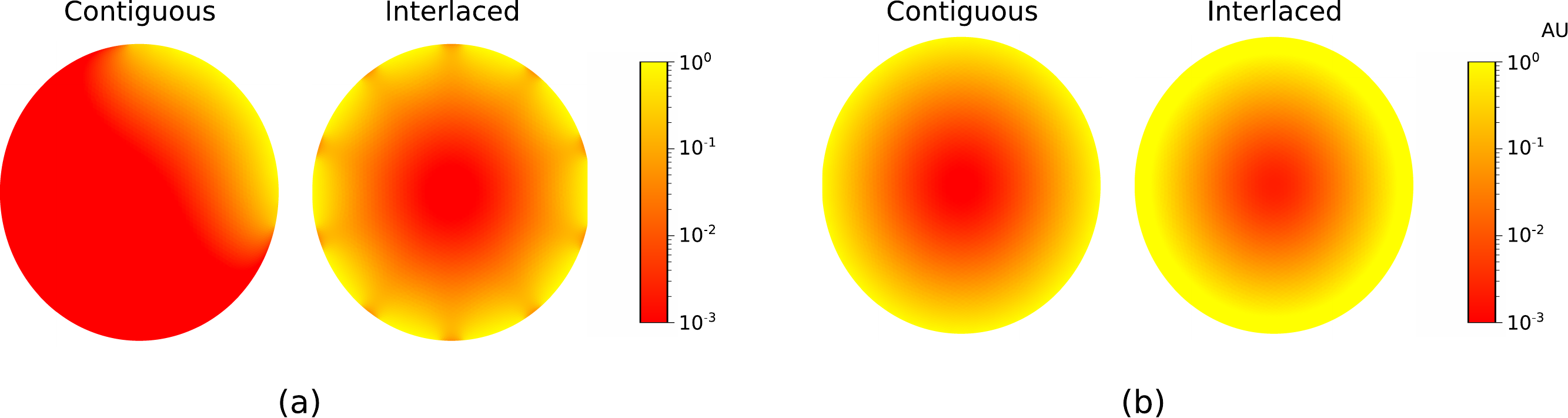}
    \caption{(a) Fluence distribution from a single illumination and (b) the sum of the fluence distributions over the four illuminations from the contiguous and interlaced design schemes. Here the fluences were computed with the absorption coefficient and reduced scattering coefficient set as the prior mean. As can be seen, both design schemes can illuminate the entire domain. }
    \label{fig:fluence_fig}
\end{figure}

\begin{figure}
    \centering
    \includegraphics[width = \textwidth]{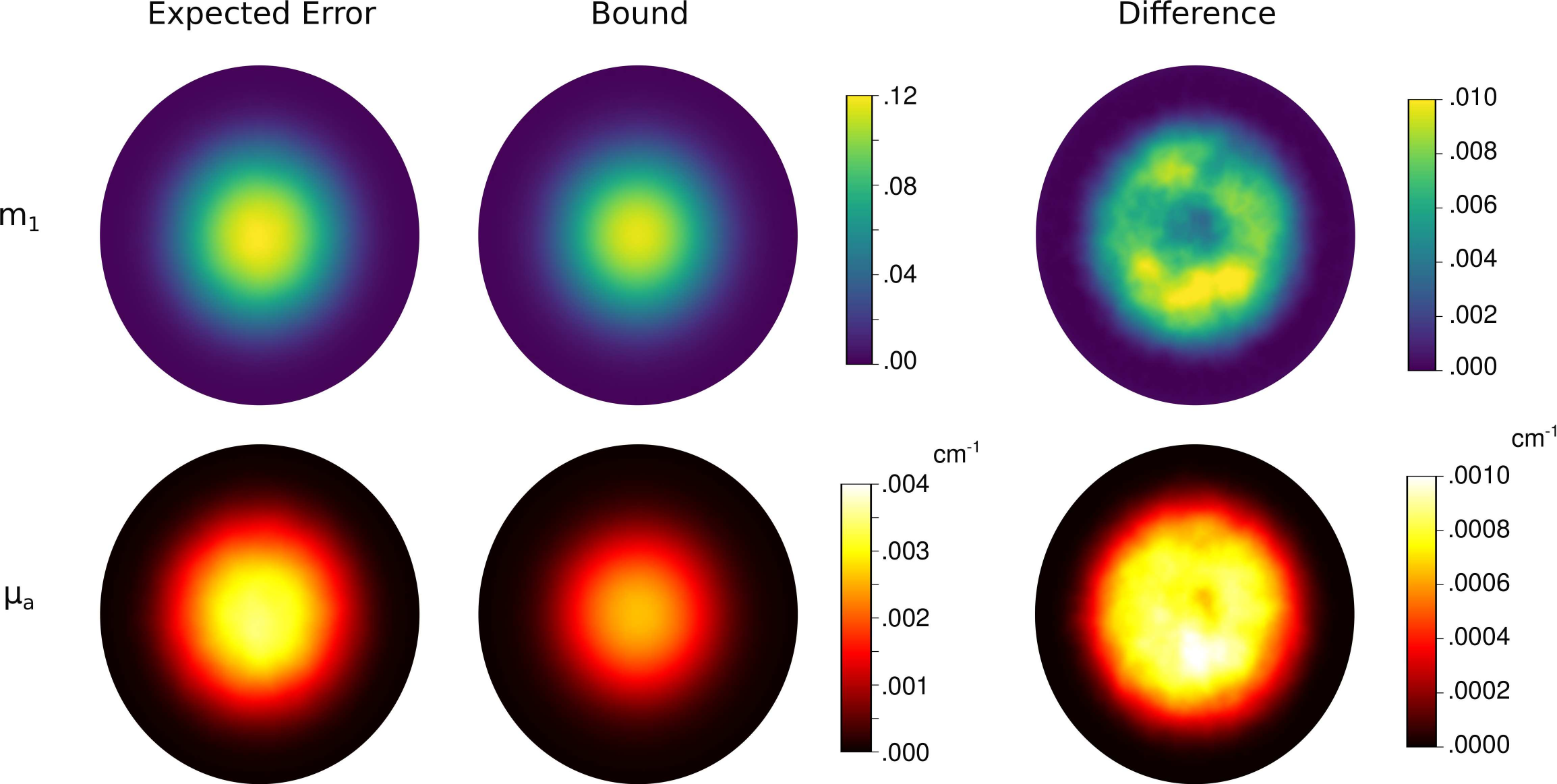}
    \caption{Study 1: Pointwise mean square error (MSE) and the corresponding Bayesian CRB bound for the parameter $m_1$ (top row) and the optical absorption coefficient $\mu_a$ (bottom row) in the single-illumination regime. The difference between the pointwise MSE and the bound is shown in the last column. As can be seen, the Bayesian CRB provides a tight lower bound on the error in $m_1$. The bound is less sharp for $\mu_a$ but still captures the general spatial structure of the error.}
    \label{fig:recon_vs_bound_val}
\end{figure}

\begin{figure}
    \centering
    \includegraphics[width = \textwidth]{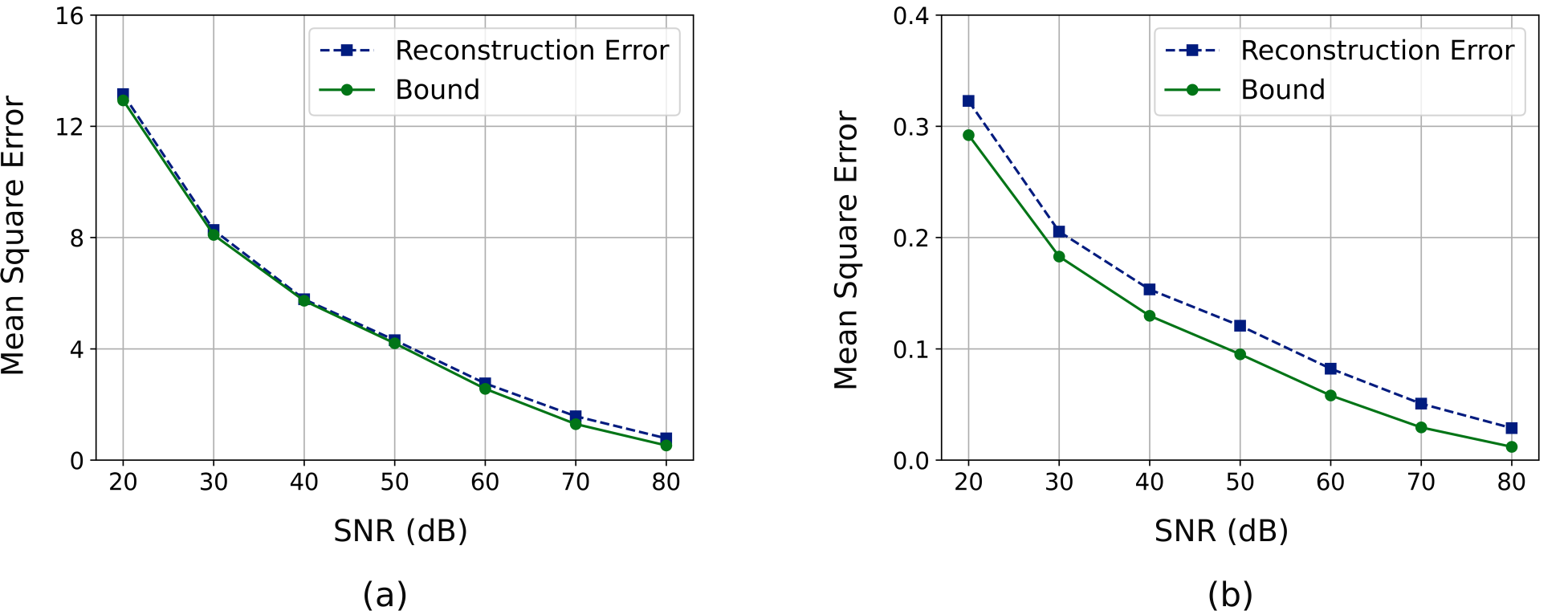}
    \caption{Study 1: The MSE and the corresponding Bayesian CRB bound as a function of signal-to-noise ratio (SNR) for the (a) the parameter $m_1$ and (b) the absorption coefficient $\mu_a$. Note that the results shown in Figures \ref{fig:mua_recon_val} and \ref{fig:recon_vs_bound_val}, which use noise variance $\sigma^2 = 10^{-3}$, correspond to an SNR value of approximately $64.46 \; \mathrm{dB}$, and that here the reconstruction error was estimated using $1000$ inverse problem solves to lessen the computational burden. The Bayesian CRB provides a very tight bound on the MSE of $m_1$ and a fairly tight lower bound on the MSE of $\mu_a$ for a wide range of SNR levels. }
    \label{fig:snr_figure}
\end{figure}

\subsection{Study 2: Known Scattering Coefficient}

We now demonstrate the performance of the proposed approach in ranking multi-illumination design schemes when the scattering coefficient is known exactly. Figure \ref{fig:mua_recon} shows reconstruction results for the contiguous and interlaced design schemes in this setting. As can be seen, both schemes capture the general features of the unknown parameter, with the interlaced scheme providing smaller reconstruction error. 

Figure \ref{fig:bound_scatknown} shows the pointwise MSEs and corresponding bounds for $m_1$ and $\mu_a$ under the contiguous and interlaced design schemes. The corresponding MSE values are given in Table \ref{tab:MSE_scatKnown}. The results show that the interlaced scheme provides lower reconstruction error for both $m_1$ and $\mu_a$. This observation is accurately predicted by the proposed approach, as the Bayesian CRB bound on the MSE in $\mu_a$ is $0.0342$ for the interlaced scheme and $0.0444$ for the contiguous scheme, indicating that the interlaced scheme is better. Further, the proposed approach is able to accurately predict the region of the domain in which the interlaced scheme provides better performance, as can be seen in the last column of the figure. In particular, the proposed approach accurately predicts that this region is roughly diamond-shaped, with the biggest gap between the two schemes being found at the corners of the diamond, which are not illuminated as strongly under the contiguous scheme (see Figure \ref{fig:design_scheme}). All of this is accomplished at a fraction of the computation cost of computing the reconstruction error: While the bound computation took under 2 hours on the Intel Gold processor, computing the reconstruction error required over 500 CPU core-hours on the supercomputer.  

\begin{figure}
    \centering
    \includegraphics[width = .85\textwidth]{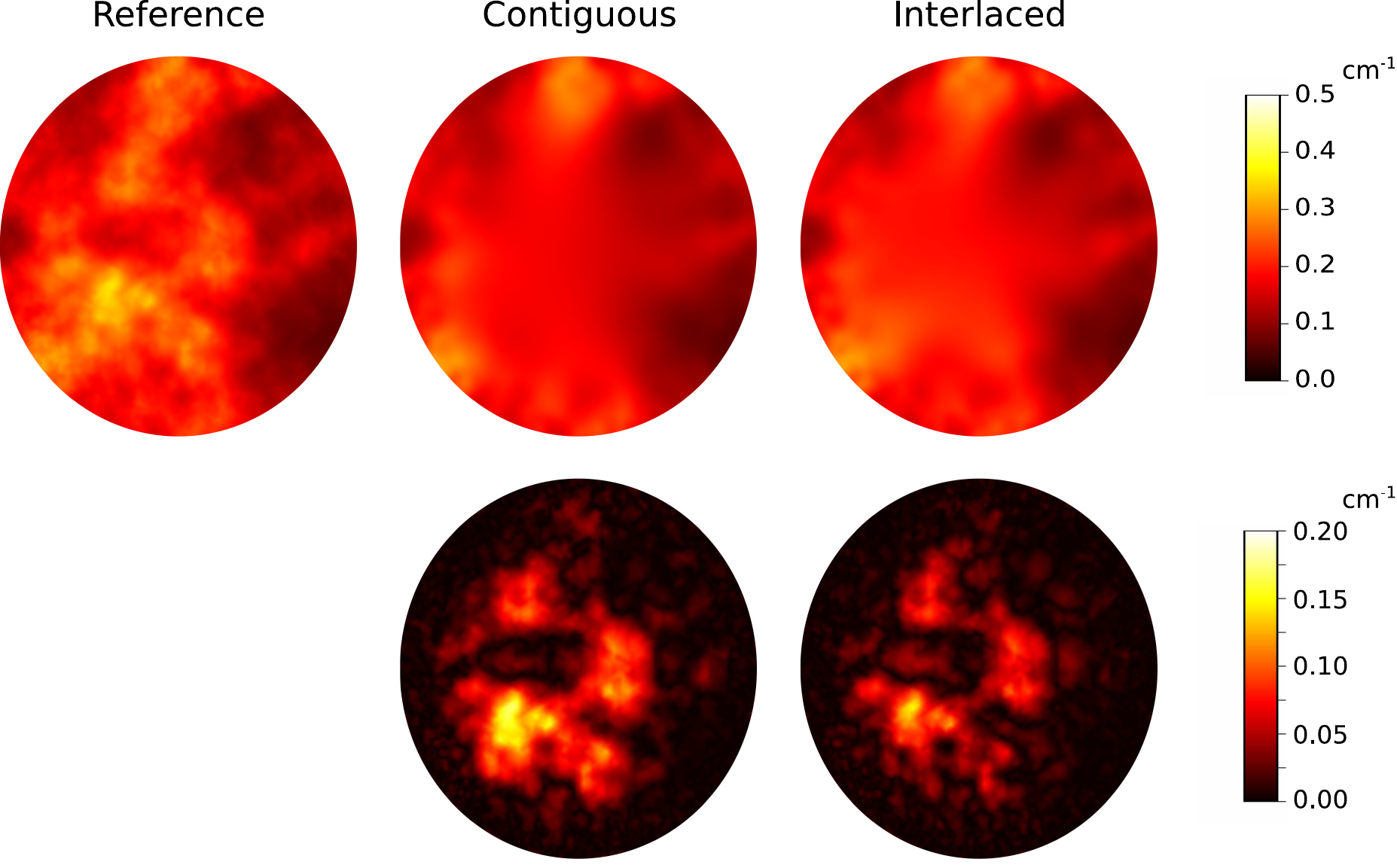}
    \caption{Study 2: Reconstructions of the absorption coefficient obtained under the contiguous and interlaced design schemes (top row) and corresponding error maps (bottom row) in the known scattering coefficient setting. The ground truth absorption coefficient is included for reference. As can be seen, both design schemes yield accurate reconstruction results near the boundary of the domain, but the interlaced scheme leads to better reconstruction performance in the interior. }
    \label{fig:mua_recon}
\end{figure}

\begin{figure}
    \centering
    \includegraphics[width = \textwidth]{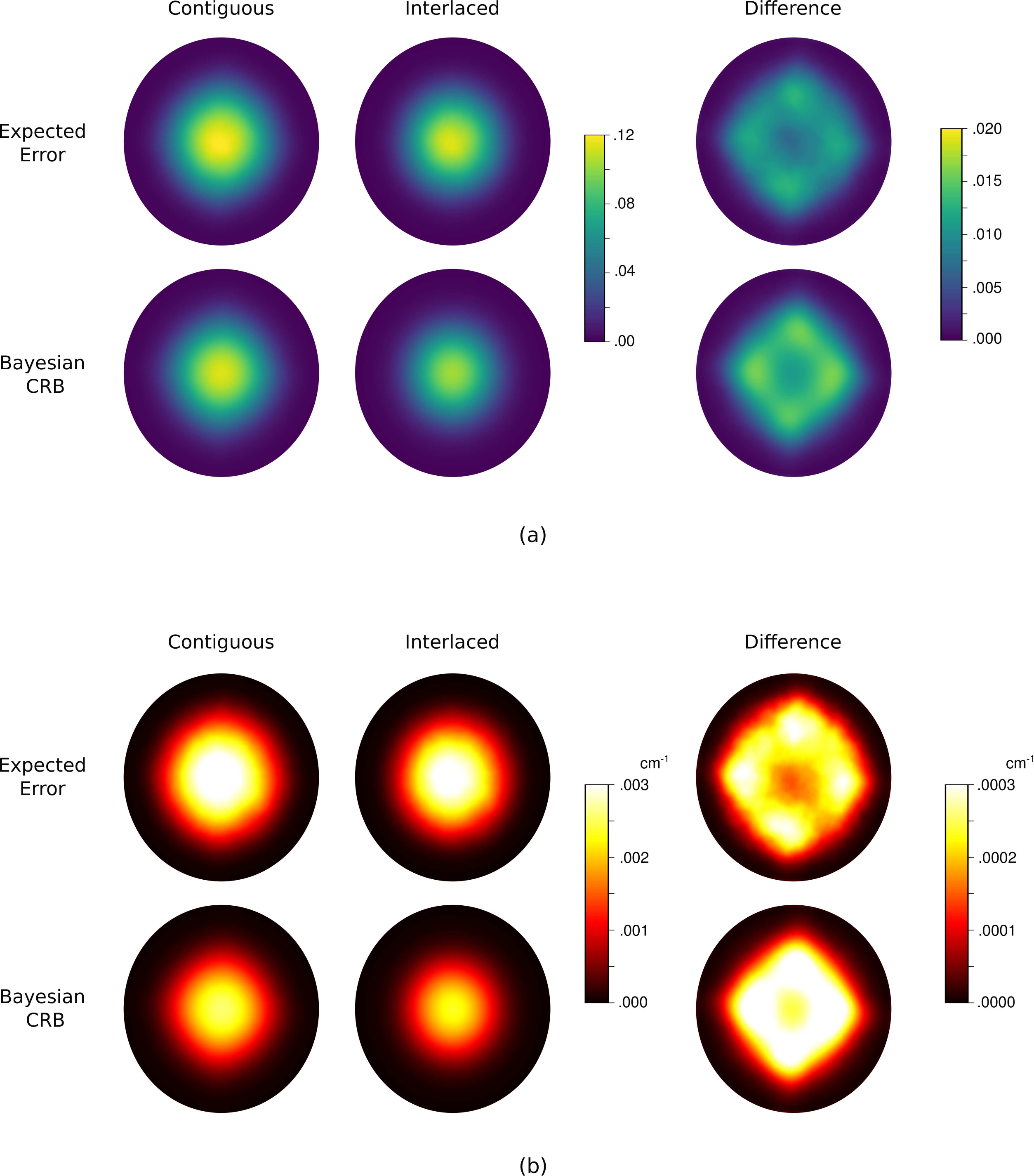}
    \caption{Study 2: Pointwise MSEs and corresponding Bayesian CRB bounds for (a) the parameter $m_1$ and (b) the optical absorption coefficient $\mu_a$. Here results are shown for both the contiguous and interlaced design schemes, with the difference between the contiguous and interlaced results displayed in the last column. As can be seen, the bound accurately predicts that the interlaced scheme will provide superior reconstruction performance and captures the spatial structure of the reconstruction error.}
    \label{fig:bound_scatknown}
\end{figure}

\begin{table}
\centering
    \begin{tabular}{c|cccc}
    \rowcolor{my_gray} &  MSE ($m_1$)   & Bound ($m_1$)   &  MSE ($\mu_a$)    & Bound ($\mu_a$)    \\ \hline
     Contiguous         &  2.1686 & 1.9580 & 0.0680 & 0.0444 \\
    Interlaced         & 1.7827 & 1.5090 & 0.0580 & 0.0342 \\ \hline
    \end{tabular}
    \vspace{0.15in}
    \caption{Study 2: The MSE and corresponding Bayesian CRB bound for the latent parameter $m_1$ and the absorption coefficient $\mu_a$ under the contiguous and interlaced design schemes in the known scattering coefficient setting. The Bayesian CRB based design metric accurately predicts the best performing design scheme, and the bound on $m_1$ is fairly tight. The bound on $\mu_a$ is less tight but still correlates with the reconstruction error. }
    \label{tab:MSE_scatKnown}
\end{table}

\subsection{Study 3: Extension to Unknown Scattering Coefficient}

Here we demonstrate the performance of the proposed approach in the setting where the reduced scattering coefficient is also unknown but is treated as a nuisance parameter. We first show representative reconstruction results. Figure \ref{fig:joint_recon} shows a single i.i.d. sample from the joint distribution of $\mu_a$ and $\mu_s'$ and the corresponding reconstructions under the contiguous and interlaced design schemes. As can be seen, the reconstruction results for $\mu_a$ are similar to those in the known scattering coefficient case, while the scattering coefficient, which, unlike the absorption coefficient, does not directly appear in the definition of the absorbed energy $h$, has slightly higher error levels. 

Figure \ref{fig:bound_scatunknown} shows the pointwise MSEs and corresponding bounds for $m_1$ and $\mu_a$ in the unknown scattering coefficient setting. The corresponding MSE summary statistics are shown in Table \ref{tab:MSE_scatunKnown}. Here we see that the bound is slightly less sharp than in the known scattering coefficient case, and unlike the known scattering coefficient case, the bound does not fully capture the structure of the difference in performance between the contiguous and interlaced schemes. However, in general, the performance here is similar to the performance in the known scattering coefficient case, demonstrating the robustness of the proposed approach. In fact, the Bayesian CRB bound on the MSE of  $\mu_a$ is $0.0481$ for the interlaced scheme and $0.0603$ contiguous scheme, which correctly predicts that the interlaced scheme will provide superior performance, and that the overall error will be higher than in the known scattering coefficient setting. Further, the computational savings are even more drastic here than in the known scattering coefficient setting: Computing the expected reconstruction error took over 3000 CPU core-hours, while computing the bound took under 4 hours.  This further demonstrates the viability of the proposed approach in guiding the design and optimization of qPACT systems. 

\begin{figure}
    \centering
    \includegraphics[width = \textwidth]{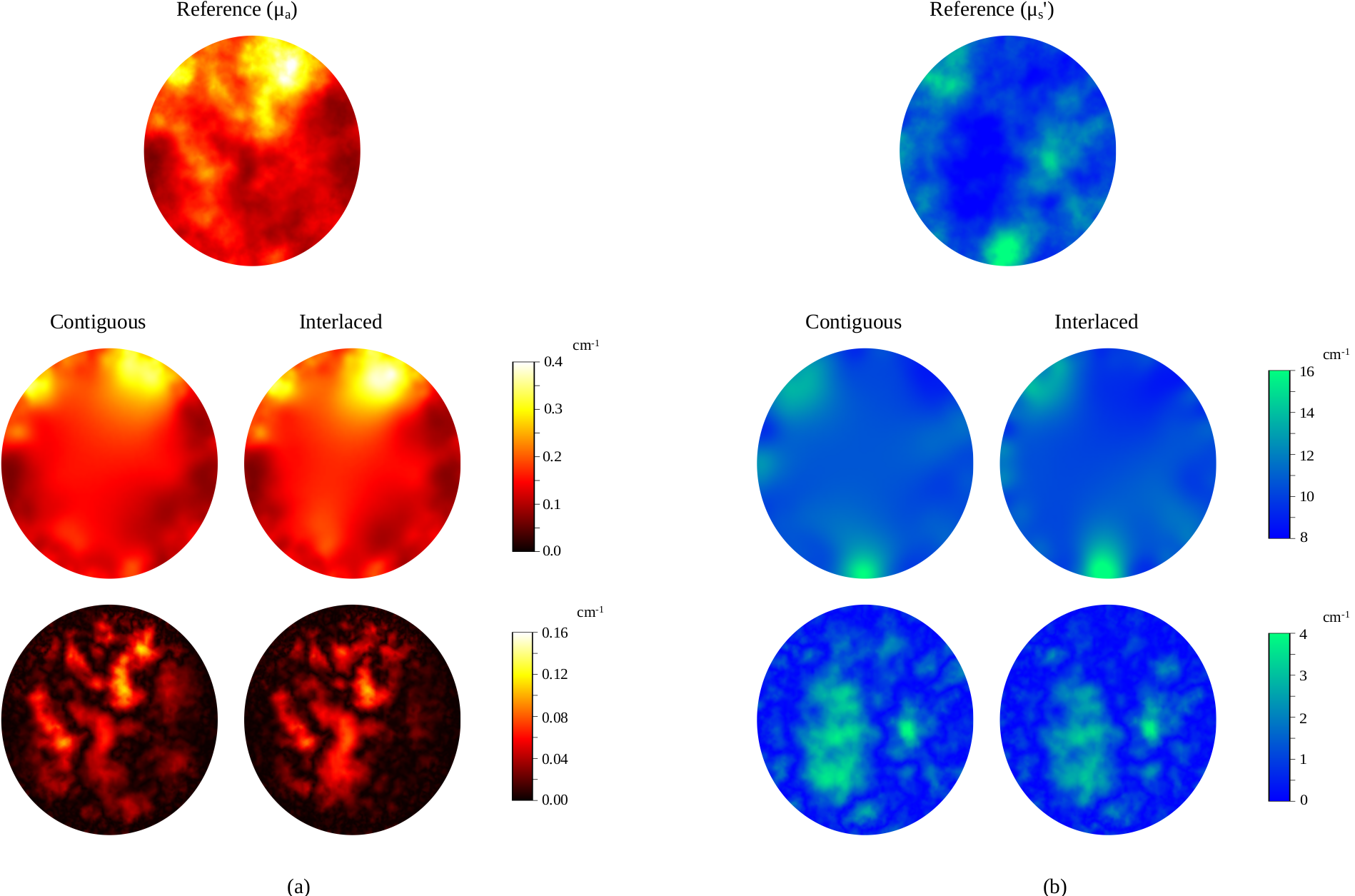}
    \caption{Study 3: Reconstructions (middle) and error maps (bottom) of the (a) absorption coefficient and (b) reduced scattering coefficient under the contiguous and interlaced design schemes. The ground truth absorption and reduced scattering coefficients (a single sample from their joint distribution) are included for reference. As can be seen, both design schemes are able to capture some key features of the unknown parameters. However, there are significant reconstruction errors, especially in the reduced scattering coefficient, due to the ill-posedness of the problem. }
    \label{fig:joint_recon}
\end{figure}

\begin{figure}
    \centering
    \includegraphics[width = \textwidth]{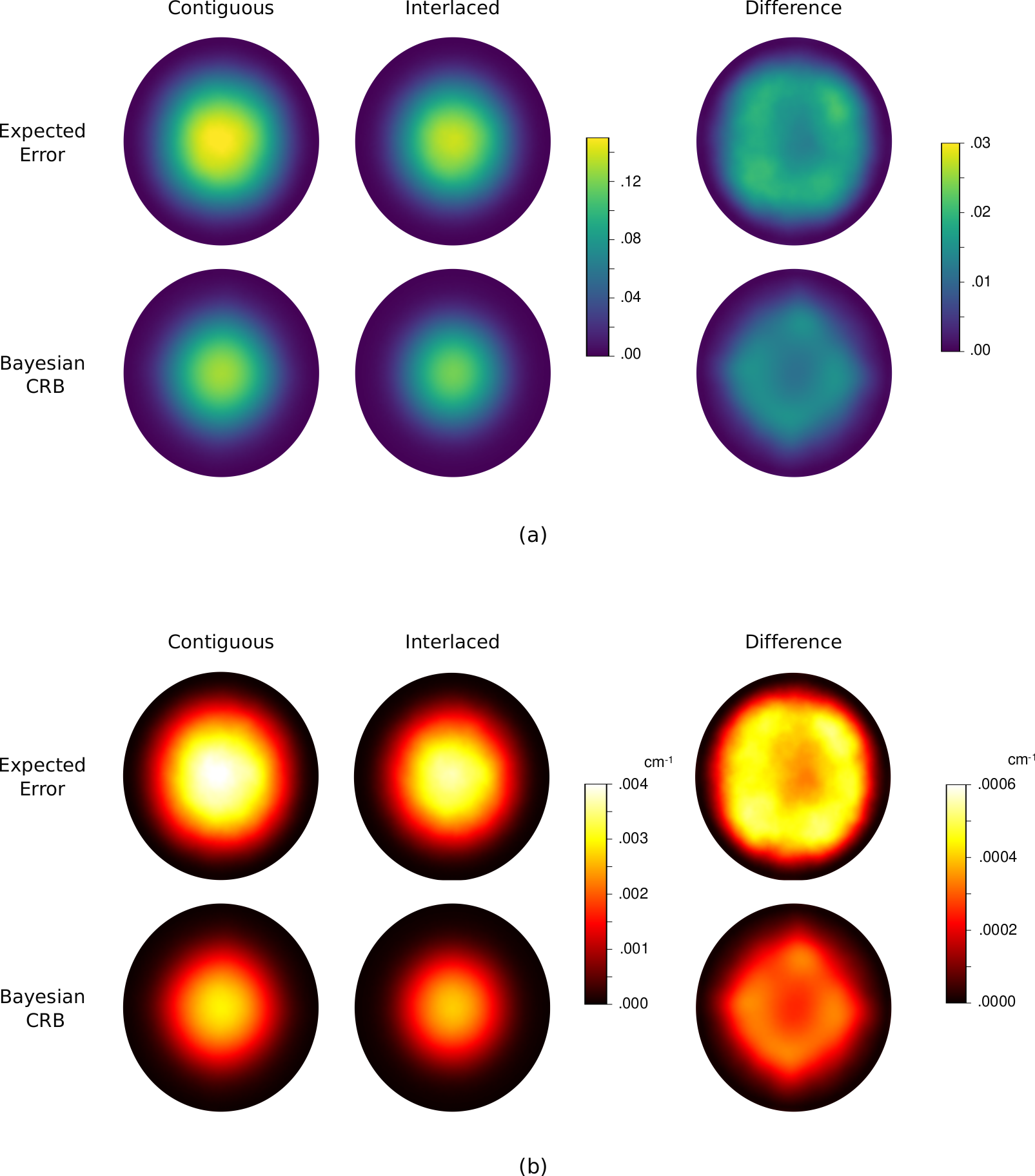}
    \caption{Study 3: Pointwise MSEs and corresponding Bayesian CRB bounds for (a) the parameter $m_1$ and (b) the corresponding optical absorption coefficient $\mu_a$ in the unknown scattering coefficient setting. Here, as in Figure \ref{fig:bound_scatknown}, results are shown for both the contiguous and interlaced design schemes, with the difference between the contiguous and interlaced results displayed in the last column. Just as in the known scattering coefficient case, the bound accurately predicts that the interlaced scheme will provide superior reconstruction performance, although the bound is less tight in this setting. }
    \label{fig:bound_scatunknown}
\end{figure}

\begin{table}
\centering
    \begin{tabular}{c|cccc}
    \rowcolor{my_gray} &  MSE ($m_1$)   & Bound ($m_1$)   &  MSE ($\mu_a$)    & Bound ($\mu_a$)    \\ \hline
     Contiguous         &  3.8458 & 2.6769 & 0.1100 & .0603 \\
    Interlaced    & 3.0168  & 2.1347 & 0.0874 &  0.0481 \\ \hline
    \end{tabular}
    \vspace{0.15in}
    \caption{Study 3: The MSE and corresponding Bayesian CRB bound for the latent parameter $m_1$ and the absorption coefficient $\mu_a$ under the contiguous and interlaced design schemes in the unknown scattering coefficient setting. While the bound in this setting is not tight, it still correlates with the reconstruction error and accurately predicts the best performing design scheme. }
    \label{tab:MSE_scatunKnown}
\end{table}

\section{Discussion}
\label{sec:discussion}

The optimal experimental design of qPACT imaging systems is challenging due to the infinite-dimensional function space setting, the presence of unknown but clinically irrelevant nuisance parameters, and the nonlinearity of the forward model. We have demonstrated that despite these challenges, the proposed OED framework can provide reliable and computationally efficient analysis of the performance of qPACT system designs and that the method is robust to nuisance parameters. Our approach relies on the Bayesian CRB to characterize the quality of a given design scheme. In particular, our approach incorporates prior distributions with trace-class covariance operators, variational adjoint based derivative computation methods, and a post-marginalization technique to address nuisance parameters. These techniques address the unique challenges of qPACT experimental design and enable efficient and accurate estimation of the bound. 


The proposed approach was analyzed using numerical studies under a stylized model of qPACT imaging systems in two space dimensions. While, as previously discussed, the two-dimensional geometry can be motivated by the quasi two-dimensionality of certain qPACT imaging systems, future work shall extend the proposed approach to three-dimensional geometries. Further, while the multi-physics forward model, which incorporates a novel model for PACT light delivery subsystems, captures several key features of the qPACT inverse problem (e.g., the presence of nuisance parameters and the non-linearity of the inverse problem), several aspects of the forward model require further investigation. These include the assumption of a lossless and acoustically homogeneous medium, that the Gr{\"u}neisen parameter is known, and the use of the diffusion approximation of the RTE to model the fluence distribution. Here the use of more sophisticated acoustic and optical models should be investigated in future work. In particular, the use of more sophisticated acoustic models that account for acoustic heterogeneity, attenuation, and dispersion (including spatially-varying and possibly unknown speed of sound and acoustic attenuation distributions) has received considerable attention in the qPACT literature in recent years \cite{Poudel_2020, matthews_2018,poudel2019survey, ranjbaran2023quantitative, kirchner2021multiple,jin2006thermoacoustic,liang2021acoustic}. While incorporating more complex physical models of wave propagation would increase the computational cost of applying the forward model, the proposed framework can be readily extended to this setting and can account for uncertainty in the acoustic properties. In particular, unknown speed of sound or acoustic attenuation distributions could be incorporated as an unknown nuisance parameter in the problem formulation \cite{GangwonVillaAnastasio24,huang2016joint} (similar to the way the reduced scattering coefficient is treated) or marginalized out in the definition of the likelihood function using the Bayesian approximation error (BAE) approach \cite{kaipio2013approximate,kaipio2006statistical,kaipio2007statistical}. Similarly, the assumption that the Gr{\"u}neisen parameter is known, which is common in the qPACT literature \cite{cox2009challenges}, could be relaxed by either incorporating the Gr{\"u}neisen parameter as an additional nuisance parameter in the inverse problem or marginalizing out in the likelihood function.

Another possible limitation of this work is the use of log-Gaussian prior distributions with covariance operators based on diffusion-reaction PDEs. 
This image prior does not capture the full complexity of biological tissues, including the correlations between the optical absorption and reduced scattering coefficients induced by the underlying anatomy and the complex geometry of vasculature networks that can be imaged with qPACT. Future research shall explore the use of more advanced priors to increase the practical relevance of the proposed work for clinical and preclinical qPACT applications. On the one hand, more sophisticated---but still ``hand-crafted''---priors could be explored to capture sharp transitions found in biological tissue, such as infinite-dimensional level-set priors \cite{iglesias2016bayesian}. On the other, generative artificial intelligence and data-driven priors could also be investigated. In particular, in future work we plan to use stochastically generated ensembles of anatomically realistic numerical breast phantoms \cite{park2023stochastic} and denoising diffusion models \cite{croitoru2023diffusion} for data-driven Bayesian CRB estimation \cite{craftsBayesian2023}.

It is also worth noting that while the Bayesian CRB provides a fairly tight lower bound on the error in estimating the latent inversion parameter, the bound becomes less tight after change-of-variable transformations and nuisance parameter post-marginalization. While we have shown that our approach is still capable of accurately predicting the design scheme with the lowest reconstruction error in the presence of these complications, future work shall investigate other bounds on the reconstruction error in Bayesian inverse problems, such as the Ziv-Zakai bound \cite{ziv1969some} or Bobrovsky-Zakai bound \cite{bobrovsky1976lower}, for use in qPACT OED. Also of interest is the extension to goal-oriented OED \cite{wu2023offline}, where the quantity of interest is a function of the qPACT parameter map (e.g., the output of a numerical observer \cite{barrett1993model}), rather than the parameter map itself.

\section{Conclusions}
\label{sec:conclusion}

In this work, we introduced a novel optimal experimental design approach for qPACT imaging systems. Our approach uses Bayesian CRB based design metrics to measure the quality of qPACT design schemes and guide their optimization. The computation of these metrics is complicated by the infinite-dimensional function space setting of qPACT and the presence of unknown but clinically irrelevant nuisance parameters in the qPACT inverse problem. In our approach, these challenges are addressed through the incorporation of prior distributions with trace-class covariance operators, variational adjoint based derivative computation methods, and a post-marginalization technique to address nuisance parameters. The resulting design metric is computationally efficient to compute and independent of the choice of estimator used to solve the qPACT inverse problem, and was shown to provide reliable and computationally efficient analysis of the performance of qPACT system designs. 



\section*{Data availability statement}
No new data were created or analysed in this study. Numerical results were generated using the open-source inverse problem library hIPPYlib available from \url{hippylib.github.io} under GNU General Public License v2.

\section*{Acknowledgements}
This work was supported by the National Institute of Biomedical Imaging and Bioengineering of the National Institutes of Health under award numbers R01EB034261 and R01EB031585.

\section*{References}

\appendix 

\section{Well-Posedness of the Bayesian CRB}
\label{sec:appendix_wellposed}

In this appendix, we introduce assumptions on the discretized forward operator $g(\bsm)$ in Eq. \eqref{eq:joint_probability} 
and the estimator $\hat{\bsm}(\bsy)$ that are sufficient to ensure well-posedness of the Bayesian CRB bound in Eq. \eqref{bound}. These assumptions are as follows. 
\begin{assumption}
\label{assumption:g_growth}
The forward model $\bsg(\bsm)$ satisfies the following conditions.

\begin{enumerate}
    \item For every $\epsilon > 0$, there exists $M = M(\epsilon)$ such that 
    $$
        \left\| \bsg(\bsm) \right\|_{\boldsymbol{\Sigma}_\bsz^{-1}} \leq \exp\{ \varepsilon \|\bsm \|^2_{\mathbf{C}^{-1}} + M\}
    $$
    for all $\bsm \in \mathbb{R}^N$. 
    \item For every $\epsilon > 0$, there exists an $K = K(\epsilon)$ such that
    $$
    \| g(\bsm_1) - g(\bsm_2) \|_{\boldsymbol{\Sigma}_\bsz^{-1}} \leq \exp\{\epsilon r^2+ K(\epsilon) \}  \| \bsm_1 - \bsm_2 \|_{\mathbf{C}^{-1}}
    $$
    for all $\bsm_1, \bsm_2 \in \mathbb{R}^N$, where $r \triangleq \max\{ \| \bsm_1 \|_{\mathbf{C}^{-1}}, \| \bsm_2 \|_{\mathbf{C}^{-1}} \}$.
\end{enumerate}
\end{assumption}

\begin{assumption}
\label{assumption:b_growth}
The conditional error $B(\bsm) \triangleq \mathbb{E}_{\bsy \mid \bsm} \left [\hat{\bsm}(\bsy) - \bsm \right]$ of the estimator $\hat{\bsm}(\bsy)$ satisfies
\begin{equation*}
\| B(\bsm) p(\bsm) \|_{\mathbf{C}^{-1}} \to 0\, \text{ as } \| \bsm \|_{\mathbf{C}^{-1}} \to \infty.
\end{equation*}
\end{assumption}

The above technical assumptions are very mild and often hold in practice.
Assumption \ref{assumption:g_growth}, which is a minor modification of Assumption 2.7 in \cite{stuart2010inverse}, can be shown to hold for a large class of forward mappings $\bsg(\bsm)$ that involve the solution of elliptic PDEs \cite{stuart2010inverse}, including the diffusion approximation of photon transport in Eq. \eqref{eq:diffuse}. Assumption \ref{assumption:b_growth} is trivially satisfied by any estimator such that $\|\hat{\bsm}(\bsy)\|_{\mathbf{C}^{-1}}$ is bounded. 

The well-posedness of the Bayesian CRB can now be proven under these assumptions.

\begin{theorem}
Under Assumptions \ref{assumption:g_growth} and \ref{assumption:b_growth}, the Bayesian CRB on the parameter $\bsm$, i.e., $\bsV_{\bsm}$, exists and satisfies the information inequality in \eqref{bound}.
\end{theorem}
\begin{proof}[Proof Sketch]

We first show that Assumption \ref{assumption:g_growth} implies that the partial derivatives $\partial p(\bsm, \bsy) / \partial \bsm_i$, $i =1 , \dots, N$, exist and are absolutely integrable. Here existence follows from the fact that since $\bsg(\bsm)$ is Lipschitz by Part (ii) of the assumption, $p(\bsm, \bsy)$ is also Lipschitz in $\bsm$, and Lipschitz functions are differentiable almost everywhere (see, e.g., Theorem 3.4.1 in \cite{cobzacs2019lipschitz}). 

Absolute integrability will be satisfied if
\begin{equation}
\label{eq:t_i}
t_i \triangleq \int_{\mathbb{R}^N} \int_{\mathbb{R}^{DI}} \left | \frac{ \partial p(\bsm, \bsy) }{\partial \bsm_i} \right | d\bsy \; d\bsm 
\end{equation}
is finite for $i=1, \dots, N$.
Using the product rule and the definition of the joint distribution $p(\bsm, \bsy)$ in Eq. \eqref{eq:joint_probability}, we obtain the following expression for the partial derivatives:
\begin{align}
\label{eq:partial}
\frac{ \partial p(\bsm, \bsy) }{\partial \bsm_i} &= \frac{\partial p(\bsy \mid \bsm) }{\partial \bsm_i} \; p(\bsm) + p(\bsy \mid \bsm) \; \frac{\partial p(\bsm)}{\partial \bsm_i} \\
\nonumber &= p(\bsm) p(\bsy \mid \bsm)  \left [ \{ \nabla \bsg(\bsm) \} \boldsymbol{\Sigma}_\bsz^{-1}(\bsy - \bsg(\bsm)) \right ]_{i} - p(\bsm) p(\bsy \mid \bsm) \left [ \bsC^{-1} \bsm \right ]_{i},
\end{align}
where $[ \boldsymbol{v} ]_i$ denotes the $i$-th component of the vector $\boldsymbol{v} \in \mathbb{R}^N$.
Substituting Eq. \eqref{eq:partial} in Eq. \eqref{eq:t_i}, we obtain
$$
t_i = \int_{\mathbb{R}^N} p(\bsm) \int_{\mathbb{R}^{DI}} p(\bsy \mid \bsm) \left | \left [ \{ \nabla \bsg(\bsm) \} \boldsymbol{\Sigma}_\bsz^{-1}(\bsy - \bsg(\bsm)) - \bsC^{-1} \bsm \right ]_{i}  \right | d\bsy \; d\bsm. 
$$

By Fernique's Theorem \cite{da2014stochastic}, the condition $t_i < +\infty$ is guaranteed to hold if, for any $\alpha > 0$, there exists a constant $C$ such that 
\begin{equation}
\label{eq:b_i}
b_i \triangleq \int_{\mathbb{R}^{DI}} p(\bsy \mid \bsm) \left | \left [\{ \nabla \bsg(\bsm) \} \boldsymbol{\Sigma}_\bsz^{-1}(\bsy - \bsg(\bsm)) - \bsC^{-1} \bsm \right ]_{i}  \right | d\bsy \leq C \, \mathrm{exp}( \alpha || \bsm ||_2^2). 
\end{equation}
Using the triangle inequality and the independence of $ \bsC^{-1} \bsm$ on $\bsy$, we have that
\begin{align}
\label{eq:b_i_triangle}
b_i &\leq \int_{\mathbb{R}^{DI}} p(\bsy \mid \bsm) \left | \left [ \bsC^{-1} \bsm \right ]_{i} \right | d \bsy + \int_{\mathbb{R}^{DI}} p(\bsy \mid \bsm) \left | \left [ \{ \nabla \bsg(\bsm) \} \boldsymbol{\Sigma}_\bsz^{-1}(\bsy - \bsg(\bsm))\right ]_{i} \right | d \bsy \nonumber \\
&= \left | \left [ \bsC^{-1} \bsm \right ]_{i} \right | + \int_{\mathbb{R}^{DI}} p(\bsy \mid \bsm) \left | \left [ \{ \nabla \bsg(\bsm) \} \boldsymbol{\Sigma}_\bsz^{-1}(\bsy - \bsg(\bsm))\right ]_{i} \right | d \bsy. 
\end{align}
We next note that 
\begin{align*}
\left | \left [ \{ \nabla \bsg(\bsm) \} \boldsymbol{\Sigma}_\bsz^{-1}(\bsy - \bsg(\bsm))\right ]_{i} \right | &\leq \left \| \{ \nabla \bsg(\bsm) \} \boldsymbol{\Sigma}_\bsz^{-1}(\bsy - \bsg(\bsm))\right  \|_2 \\
&\leq \left \| \nabla \bsg(\bsm) \right \|_2 \left \| \boldsymbol{\Sigma}_\bsz^{-1} \right \|_2 \left \| \bsy - \bsg(\bsm) \right \|_2,
\end{align*}
so we have that 
\begin{equation}
\label{eq:b_almost}
b_i \leq \left | \left \{ \bsC^{-1} \bsm \right \}_{i} \right | + \left \| \nabla \bsg(\bsm) \right \|_2 \left \| \boldsymbol{\Sigma}_\bsz^{-1} \right \|_2 \int_{\mathbb{R}^{DI}} p(\bsy \mid \bsm) \left \| \bsy - \bsg(\bsm) \right \|_2 d \bsy.
\end{equation}
The integral in the above expression is finite and independent of $\bsm$. What remains is to bound $ \left \| \nabla \bsg(\bsm) \right \|_2$. We now claim that there exists a constant $C_0$ (independent of $\bsm$) such that
\begin{equation}
\label{eq:jac_bound}
\| \nabla g(\bsm) \|_{2} \leq C_0 \exp\{\epsilon \| \bsm \|_2^2+ K(\epsilon) \}
\end{equation}
for all $\epsilon > 0$, where $K(\epsilon)$ is defined by Assumption \ref{assumption:g_growth}. This claim can be verified using the fact that Part (ii) of Assumption \ref{assumption:g_growth} implies Lipschitzness of the partial derivatives of $\bsg(\bsm)$ (see, e.g., Proposition 2.2.1 in \cite{cobzacs2019lipschitz}), together with the equivalence of finite-dimensional norms. Substituting the bound in \eqref{eq:jac_bound} into \eqref{eq:b_almost}, setting $\epsilon = \alpha$, and using the fact that linear functions can be bounded by exponential functions, we can verify that \eqref{eq:b_i} holds. The partial derivatives of $p(\bsm, \bsy)$ are therefore absolutely integrable. 

We have shown that the partial derivatives exist and are absolutely integrable. Using the proof techniques in \cite{bell_2013_detection} (in particular, the proof of Eq. 4.520), the well-posedness of the bound then follows from the absolute integrability of the partials and Assumption \ref{assumption:b_growth}.
\end{proof}

\section{Derivative Computation with the Variational Adjoint Method}
\label{sec:appendix_deriv}

In this appendix, we derive an expression for the score of the qPACT likelihood function in \eqref{eq:Jd_estimator}. Specifically, this appendix uses the variational adjoint method \cite{troltzsch2010optimal} to obtain an expression for the continuous-valued qPACT likelihood score $\nabla_{\bm{m}} \log p(\bsy \mid \bm{m})$, from which $\nabla_{\bsm} \log p(\bsy \mid \bsm)$ can be obtained via discretization. Throughout, we assume that all functions are square-integrable and have square-integrable derivatives, i.e., all functions are in the Sobolev space $H^{1}(\Omega)$. 

We begin by deriving the weak form of the diffusion PDE in \eqref{eq:diffuse}, which is restated below for convenience: 
\begin{equation}
\left\{
\begin{array}{lr}
     \mu_a(\bsx) \phi(\bsx) - \nabla \cdot D(\bsx) \nabla \phi(\bsx) = 0 & \text{for any} \; \bsx \in \Omega, \\[1mm] 
     \left \langle D(\bsx)  \nabla \phi(\bsx), \bseta (\bsx) \right \rangle + \frac{1}{2}\phi(\bsx) = 2q (\bsx) & \text{for any} \; \bsx \in \partial \Omega.
     \label{eq:diffuse_appendix}
\end{array}
\right.
\end{equation}
To that end, we first multiply the system by a test function $\tilde{\phi}(\bsx)$ and integrate over the domain to obtain
\begin{equation}
    \int_{\Omega} \mu_a(\bsx) \phi(\bsx) \tilde{\phi}(\bsx) - (\nabla \cdot D(\bsx) \nabla \phi(\bsx)) \tilde{\phi}(\bsx) \; d\bsx = 0, \label{eq:strong_test_1}  
\end{equation}
for any $\tilde{\phi}(\bsx)$. Integrating by parts the second term in \eqref{eq:strong_test_1} and using the boundary condition in \eqref{eq:diffuse_appendix} gives
\begin{align*}
 &\int_{\Omega} (\nabla \cdot D(\bsx) \nabla \phi(\bsx)) \tilde{\phi}(\bsx) \; d\bsx \\
 & \hspace{3cm} =   \int_{\partial \Omega}  D(\bsx) \left \langle \nabla \phi(\bsx), \bseta(\bsx) \right \rangle \tilde{\phi}(\bsx) \; d\bsx - \int_{\Omega}  \left \langle D(\bsx) \nabla \phi(\bsx), \nabla \tilde{\phi}(\bsx) \right \rangle \; d \bsx \\
 & \hspace{3cm} =   \frac{1}{2} \int_{\partial \Omega} q(\bsx) \tilde{\phi}(\bsx) - \phi(\bsx) \tilde{\phi}(\bsx) \; d\bsx - \int_{\Omega}  \left \langle D(\bsx) \nabla \phi(\bsx), \nabla \tilde{\phi}(\bsx)\right \rangle \; d \bsx.
\end{align*}
Inserting the above expression into \eqref{eq:strong_test_1}, we obtain
\begin{align}
    &\int_{\Omega} \mu_a(\bsx) \phi(\bsx) \tilde{\phi}(\bsx) + \left \langle D(\bsx) \nabla \phi(\bsx), \nabla \tilde{\phi}(\bsx) \right \rangle \; d \bsx + \frac{1}{2} \int_{\partial \Omega} \phi(\bsx) \tilde{\phi}(\bsx) - q(\bsx) \tilde{\phi}(\bsx)  \; d\bsx = 0 
\label{eq:weak_form}
\end{align}
for any $\tilde{\phi}(\bsx) \in H^{1}(\Omega)$. The above equation is the weak form of \eqref{eq:diffuse}.

We now form a Lagrangian function $\mathcal{L}$ that combines the log-likelihood function with the weak form of the diffusion PDE. Here for simplicity we assume that there is only $I = 1$ illumination (the generalization to multiple illuminations is trivial). The Lagrangian is therefore
\begin{align*}
\mathcal{L}\left ( \{\phi, h\}, \{\tilde{\phi}, \tilde{h} \}, \{ m_1, m_2 \} \right) &\triangleq \frac{1}{2}||\mathcal{H}_{a} h - \bsd ||_{\boldsymbol{\Sigma}_{\mathbf{z}}^{-1}}^2  + 
\int_{\Omega} \mu_a \, \phi \, \tilde{\phi} + \left \langle \frac{1}{3(\mu_a + \mu'_s)} \nabla \phi, \nabla \tilde{\phi} \right \rangle \; d \bsx \\
& \quad + \int_{\Omega} h\tilde{h} - \phi \mu_a \tilde{h} \; d\bsx + \frac{1}{2} \int_{\partial \Omega} \phi \tilde{\phi} - q \tilde{\phi}  \; d\bsx,
\end{align*}
where $\mu_a(\bsx) = \overline{\mu_a} e^{m_1(\bsx)}$ and $\mu_s'(\bsx) = \overline{\mu_s'} e^{m_2(\bsx)}$, and the condition $h(\bsx) = \mu_a(\bsx) \phi(\bsx)$ is enforced weakly using the test function $\tilde{h}(\bsx)$.

The score of the likelihood function can be obtained as the first variation of the Lagrangian with respect to the parameter $\bm m$ when the variations with respect to the state $\{\phi, h\}$ and adjoint variable $\{\tilde{\phi, \tilde{h}}\}$ are set to zero. Requiring variations with respect to $\tilde{\phi}$ to vanish recovers the weak form of the PDE, while requiring variations with respect to $\tilde{h}$ to vanish yields the weak form of the equation $h(\bsx) = \mu_a(\bsx) \phi(\bsx)$. Requiring variations with respect to $\phi$ to vanish yields the following equation: 
\begin{equation}
\label{eq:pre_weak}
\int_{\Omega} \left [ \mu_a \hat{\phi} \tilde{\phi} + \left \langle \frac{1}{3(\mu_a + \mu'_s)} \nabla \hat{\phi}, \nabla \tilde{\phi} \right \rangle - \hat{\phi} \mu_a \tilde{h} \right]  d \bsx + \frac{1}{2} \int_{\partial \Omega} \hat{\phi} \tilde{\phi}  \; d\bsx = 0,
\end{equation}
for any variation $\hat{\phi} \in H^1(\Omega)$ in $\phi$. Requiring variations with respect to $h$ to vanish gives the condition
$$
\left \langle \mathcal{H}_{a} h - \bsd, \mathcal{H}_{a} \hat{h} \right \rangle_{\boldsymbol{\Sigma}_{\mathbf{z}}^{-1}} = \int_{\Omega} \hat{h} \tilde{h}.
$$
for any variation $\hat{h}(\bsx) \in H^1(\Omega)$ of $h$. Since $\hat{h}$ is arbitrary, this implies that 
\begin{equation}
\tilde{h} = \mathcal{H}_a^* \boldsymbol{\Sigma}_{\mathbf{z}}^{-1}( \mathcal{H}_{a} h - \bsd) 
\label{eq:g_form_appendix}
\end{equation}
where $\mathcal{H}_a^*$ is the adjoint of $\mathcal{H}_a$. Inserting the above expression into \eqref{eq:pre_weak}, we then have that 
$$
\int_{\Omega} \left [ \mu_a \hat{\phi} \tilde{\phi} + \left \langle \frac{1}{3(\mu_a + \mu'_s)} \nabla \hat{\phi}, \nabla \tilde{\phi} \right \rangle - \hat{\phi} \mu_a  \mathcal{H}_a^* \boldsymbol{\Sigma}_{\mathbf{z}}^{-1}( \mathcal{H}_{a} h - \bsd)  \right]  d \bsx + \frac{1}{2} \int_{\partial \Omega} \hat{\phi} \tilde{\phi}  \; d\bsx = 0.
$$
Integrating by parts the second term in the above equation, we obtain 
\begin{align*}
&\int_{\Omega} \left [ \mu_a \hat{\phi} \tilde{\phi} + \hat{\phi} \nabla \cdot \left (\frac{1}{3(\mu_a + \mu'_s)}\nabla \tilde{\phi} \right) - \hat{\phi} \mu_a \mathcal{H}_a^* \boldsymbol{\Sigma}_{\mathbf{z}}^{-1}( \mathcal{H}_{a} h - \bsd)  \right]  d \bsx \\
&\quad - \int_{\partial \Omega} \hat{\phi} \frac{1}{3(\mu_a + \mu'_s)} (\nabla \tilde{\phi} \cdot \bseta) \; d\bsx  + \frac{1}{2} \int_{\partial \Omega} \hat{\phi} \tilde{\phi}  \; d\bsx = 0.
\end{align*}
Making arguments regarding the arbitrariness of $\hat{\phi}$, we then obtain the \textit{strong form of the adjoint equation} for this inverse problem, which can be written as 
\begin{align*}
\mu_a \tilde{\phi} + \nabla \cdot \left (\frac{1}{3(\mu_a + \mu'_s)}\nabla\tilde{\phi} \right) -  \mu_a \mathcal{H}_a^* \boldsymbol{\Sigma}_{\mathbf{z}}^{-1}( \mathcal{H}_{a} h - \bsd) = 0 \quad &\text{in } \Omega, \\
 \frac{1}{2} \tilde{\phi} - \frac{1}{3(\mu_a + \mu'_s)} (\nabla \tilde{\phi} \cdot \bseta) = 0 \quad  &\text{on } \partial \Omega. 
\end{align*}
This is a linear PDE with a Robin boundary condition. 

We are now in a position to derive the score of the likelihood function by considering variation of the Lagrangian with respect to $m_1(\bsx)$ and $m_2(\bsx)$.  Specifically, we have that that the Fr{\'e}chet derivative of the log likelihood in the direction $\hat{m}_1(\bsx)$ is the variation of the Lagrangian in the $m_1$ direction, i.e., 
$$
\int_{\Omega} \left [ \dot{\mu}_a \phi \tilde{\phi} \hat{m}_1 - \left \langle \frac{\dot{\mu}_a \hat{m}_1}{3(\mu_a + \mu_s')^2} \nabla \phi, \nabla \tilde{\phi} \right \rangle  - \phi \dot{\mu}_a \tilde{h} \hat{m}_1 \right ] \; d\bsx,
$$
where $\tilde{h}$ is given by \eqref{eq:g_form_appendix}, $\tilde{\phi}$ is obtained by solving the adjoint PDE, and $\dot{\mu}_a$ is the derivative of the pointwise transformation $\mu_a = \overline{\mu_a} e^{m_1}$, i.e., $\dot{\mu}_a = \mu_a$.  Let $\mathcal{G}_1(\bm{m})$ denote the gradient of the log likelihood with respect to $m_1$, defined as the Riesz representer of the Fr{\'e}chet derivative, evaluated at the point $\bm{m} = [m_1, m_2]^T$. Since $\hat{m}_1$ is arbitrary, we have that
$$
\mathcal{G}_1(\bm{m}) = \begin{cases} \dot{\mu}_a \phi \tilde{\phi} - \left \langle \frac{\dot{\mu}_a}{3(\mu_a + \mu_s')^2} \nabla \phi, \nabla \tilde{\phi} \right \rangle  - \phi \dot{\mu}_a \tilde{h}, \quad \text{in} \; \Omega, \\
0 \hspace{6.55cm} \text{on} \; \partial \Omega.
\end{cases}
$$
Similarly, the Fr{\'e}chet derivative of the log likelihood in the direction $\hat{m}_2(\bsx)$ is the variation of the Lagrangian in the $m_2$ direction, i.e., 
$$
\int_{\Omega} \left \langle -\frac{\dot{\mu}_s' \hat{m}_2}{3(\mu_a + \mu_s')^2} \nabla \phi, \nabla \tilde{\phi} \right \rangle  \; d\bsx,
$$
where $\dot{\mu}_s'$ is the pointwise derivative of $\mu_s' = \overline{\mu_s'} e^{m_2}$, i.e., $\dot{\mu}_s' = \mu_s'$. The corresponding gradient is
$$
G_2(\bm{m}) = \begin{cases} \left \langle -\frac{\dot{\mu}_s' \hat{m}_2}{3(\mu_a + \mu_s')^2} \nabla \phi, \nabla \tilde{\phi} \right \rangle   \quad \text{in} \; \Omega, \\
0 \hspace{4cm} \text{on} \; \partial \Omega. 
\end{cases}
$$
Noting that $\nabla_{\bm{m}} \log p(\bsy \mid \bm{m}) = [G_1(\bm{m}), G_2(\bm{m})]^T$ completes the derivation. 
\end{document}